\documentclass[a4paper,11pt]{article}

\usepackage[utf8]{inputenc}
\usepackage{palatino}
\usepackage[margin=3cm]{geometry}
\usepackage{amsmath}
\usepackage{amssymb}
\usepackage{amsthm}  
\usepackage{mathabx}
\usepackage{stmaryrd}
\usepackage[colorlinks=true,linkcolor=black,citecolor=black,plainpages=false,pdfpagelabels]{hyperref}
\hypersetup{pdftitle={One-Shot Decoupling}}
\usepackage{bbm}

\usepackage{calc}  % used for the drawing of arrows 
\usepackage[small,bf]{caption}

\newtheorem{thm}{Theorem}[section]
\newtheorem{lem}[thm]{Lemma}
\newtheorem{cor}[thm]{Corollary}

\newtheorem{defin}[thm]{Definition}
\newcommand{\bra}[1]{\langle #1|}
\newcommand{\ket}[1]{|#1 \rangle}
\newcommand{\braket}[2]{\langle #1|#2\rangle}
\newcommand{\ketbra}[1]{\ket{#1}\bra{#1}}
\newcommand{\proj}[1]{\ket{#1}\bra{#1}}

\newcommand{\ident}{\mathbbm{1}}
\DeclareMathOperator{\tr}{\mathrm{Tr}}

\newcommand{\mdag}{^{\dag}} % dag operator

\newcommand{\demi}{\frac{1}{2}}

\DeclareUnicodeCharacter{0393}{\Gamma}
\DeclareUnicodeCharacter{0394}{\Delta}
\DeclareUnicodeCharacter{0398}{\Theta}
\DeclareUnicodeCharacter{039B}{\Lambda}
\DeclareUnicodeCharacter{039E}{\Xi}
\DeclareUnicodeCharacter{03A0}{\Pi}
\DeclareUnicodeCharacter{03A3}{\Sigma}
\DeclareUnicodeCharacter{03A6}{\Phi}
\DeclareUnicodeCharacter{03A8}{\Psi}
\DeclareUnicodeCharacter{03A9}{\Omega}
\DeclareUnicodeCharacter{03B1}{\alpha}
\DeclareUnicodeCharacter{03B2}{\beta}
\DeclareUnicodeCharacter{03B3}{\gamma}
\DeclareUnicodeCharacter{03B4}{\delta}
\DeclareUnicodeCharacter{03B5}{\varepsilon}
\DeclareUnicodeCharacter{03B6}{\zeta}
\DeclareUnicodeCharacter{03B7}{\eta}
\DeclareUnicodeCharacter{03B8}{\theta}
\DeclareUnicodeCharacter{03BA}{\kappa}
\DeclareUnicodeCharacter{03BB}{\lambda}
\DeclareUnicodeCharacter{03BC}{\mu}
\DeclareUnicodeCharacter{03BD}{\nu}
\DeclareUnicodeCharacter{03BE}{\xi}
\DeclareUnicodeCharacter{03C0}{\pi}
\DeclareUnicodeCharacter{03C1}{\rho}
\DeclareUnicodeCharacter{03C3}{\sigma}
\DeclareUnicodeCharacter{03C4}{\tau}
\DeclareUnicodeCharacter{03C6}{\varphi}
\DeclareUnicodeCharacter{03C7}{\chi}
\DeclareUnicodeCharacter{03C8}{\psi}
\DeclareUnicodeCharacter{03C9}{\omega}

\DeclareUnicodeCharacter{2297}{\otimes}
\DeclareUnicodeCharacter{2A02}{\otimes}
\DeclareUnicodeCharacter{2264}{\leqslant}
\DeclareUnicodeCharacter{2265}{\geqslant}
\DeclareUnicodeCharacter{B2}{^2}
\DeclareUnicodeCharacter{B3}{^3}
\DeclareUnicodeCharacter{221A}{\sqrt}
\DeclareUnicodeCharacter{BD}{\frac{1}{2}}
\DeclareUnicodeCharacter{2102}{\mathbb{C}}
\DeclareUnicodeCharacter{2208}{\in}
\DeclareUnicodeCharacter{2211}{\sum}
\DeclareUnicodeCharacter{222B}{\int}
\DeclareUnicodeCharacter{2329}{\langle}
\DeclareUnicodeCharacter{232A}{\rangle}
\DeclareUnicodeCharacter{0418}{\ident}   % И
\DeclareUnicodeCharacter{026A}{\ident}   % ɪ
\DeclareUnicodeCharacter{210B}{\mathcal{H}}
\DeclareUnicodeCharacter{2110}{\mathcal{I}}
\DeclareUnicodeCharacter{2112}{\mathcal{L}}
\DeclareUnicodeCharacter{211B}{\mathcal{R}}
\DeclareUnicodeCharacter{212C}{\mathcal{B}}
\DeclareUnicodeCharacter{2130}{\mathcal{E}}
\DeclareUnicodeCharacter{2131}{\mathcal{F}}
\DeclareUnicodeCharacter{2133}{\mathcal{M}}
\DeclareUnicodeCharacter{2190}{\leftarrow}
\DeclareUnicodeCharacter{2192}{\rightarrow}
\DeclareUnicodeCharacter{2020}{^\dagger}
\DeclareUnicodeCharacter{2016}{\|}
\DeclareUnicodeCharacter{2A01}{\varoplus}
\DeclareUnicodeCharacter{203E}{\bar}
\DeclareUnicodeCharacter{22A4}{^\top} % transpose
\DeclareUnicodeCharacter{2200}{\forall}
\DeclareUnicodeCharacter{2203}{\exists}
\DeclareUnicodeCharacter{2260}{\neq}
\DeclareUnicodeCharacter{2282}{\subset}
\DeclareUnicodeCharacter{2286}{\subseteq}
\DeclareUnicodeCharacter{2209}{\notin}
\DeclareUnicodeCharacter{2261}{\equiv}
\DeclareUnicodeCharacter{2205}{\emptyset}

\newcommand{\mbU}{\mathbb{U}}

\newcommand{\mbE}{\mathbb{E}}

\newcommand*{\cB}{\mathcal{B}}

\newcommand*{\cF}{\mathcal{F}}

\newcommand*{\cH}{\mathcal{H}}

\newcommand*{\cL}{\mathcal{L}}

\newcommand*{\cP}{\mathcal{P}}
\newcommand*{\cS}{\mathcal{S}}
\newcommand*{\cU}{\mathcal{U}}
\newcommand*{\cT}{\mathcal{T}}
\newcommand*{\cV}{\mathcal{V}}
\newcommand*{\cW}{\mathcal{W}}

\DeclareMathOperator{\Span}{span}

\renewcommand{\otimes}{\varotimes}

\newcommand*{\eps}{\varepsilon}

\newcommand*{\opid}{\mathcal{I}}

\newcommand{\hmin}{H_{\min}}
\newcommand{\hmax}{H_{\max}}
\newcommand{\tps}{\intercal}

\newcommand*{\pr}[3]{\put(#1,#2){\makebox(0,0){#3}}}

\newcommand*{\arrowwidth}{\linethickness{1pt}}

\newcommand*{\arrowrightp}{\vspace{0ex}\hspace{-0.5ex}$\blacktriangleright$}

\newcounter{tempx}\newcounter{tempy}

\newcommand*{\arrowr}[3]{\arrowwidth
                        \setcounter{tempx}{#3-5}
                        \put(#1,#2){\line(1,0){\value{tempx}}}            
                        \put(#1,#2){\makebox(#3,0)[r]{ \arrowrightp  }}
}

\begin{document}
\title{One-Shot Decoupling}

\author{Fr\'ed\'eric Dupuis${}^\ast$
\quad\quad
Mario Berta${}^\ast$
\quad\quad
J\"urg Wullschleger${}^{\dagger,\ddagger}$
\quad\quad
Renato Renner${}^\ast$\\[4mm]
${}^\ast$%
{\it\small Institute for Theoretical Physics}\\[-1mm]
{\it\small ETH Zurich, Switzerland}\\[2mm]
${}^\dagger$%
{\it\small Department of Computer Science and Operations Research}\\[-1mm]
{\it\small Universit\'e de Montr\'eal, Quebec, Canada}\\[2mm]
${}^\ddagger$%
{\it\small McGill University, Quebec, Canada}\\[2mm]
}

\date{}

\maketitle

\begin{abstract}
If a quantum system $A$, which is initially correlated to another system, $E$, undergoes an evolution separated from $E$, then the correlation to $E$ generally decreases.  Here, we study the conditions under which the correlation disappears (almost) completely, resulting in a decoupling of $A$ from $E$. We give a criterion for decoupling in terms of two smooth entropies, one quantifying the amount of initial correlation between $A$ and $E$, and the other characterizing the mapping that describes the evolution of $A$. The criterion applies to arbitrary such mappings in the general one-shot setting. Furthermore, the criterion is tight for mappings that satisfy certain natural conditions. One-shot decoupling has a number of applications both in physics and information theory, e.g., as a building block for quantum information processing protocols. As an example, we give a one-shot state merging protocol and show that it is essentially optimal in terms of its entanglement consumption/production.
\end{abstract}

%\tableofcontents

\section{Introduction} \label{sec:Intro}

Correlations in quantum systems, and in particular entanglement, have
been in the focus of (both theoretical and experimental) research in
quantum information science over the past decades. As a result, one
has nowadays a pretty good (although still not complete) understanding
of quantum correlations and, in particular, the processes that create
them. In this work, we take | so to speak | an opposite approach and study
conditions under which two systems can be decoupled, i.e., brought to
a state where they are uncorrelated.

We call a system, $B$, decoupled from another system, $E$, if
the joint state of the two systems, $\rho_{B E}$, has product form
$\rho_{B} \otimes \rho_{E}$. Operationally, this means that the
outcome of any measurement on $B$ is statistically independent of the
outcome of any measurement on $E$. Or, in information-theoretic terms,
the system $E$ does not give any information on $B$ (and can therefore
safely be ignored when studying~$B$).

\paragraph{Decoupling Theorem.} Our goal is to characterize the conditions under which the evolution of a system results in decoupling. For this, we consider a system, $A$, that may initially be correlated to $E$. Furthermore, we assume that the system $A$ undergoes an evolution, described by a TPCPM\footnote{A trace-preserving completely-positive map (TPCPM) is a linear function that maps density operators to density operators.} $\mathcal{\bar{T}}$ from $A$ to $B$, during which no interaction with $E$ takes place (see Fig.~\ref{fig:decoupling}). The main result of this work is a decoupling theorem, i.e., a criterion that provides necessary and sufficient conditions for decoupling (of $B$ from $E$). The criterion depends on two entropic quantities, characterizing the initial state, $\rho_{A E}$, and the mapping $\mathcal{\bar{T}}$, respectively.

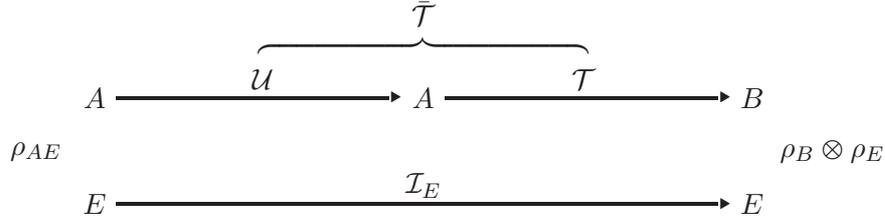
\begin{figure}
\begin{center}
\begin{picture}(250,85)(0, 0)
  \arrowr{10}{10}{230}
  \arrowr{10}{50}{107}
  \arrowr{133}{50}{107}
  \pr{65}{57}{$\mathcal{U}$}  
   \pr{185}{57}{$\cT$}  
  \pr{125}{17}{$\opid_E$}
  \pr{125}{70}{$\overbrace{\hphantom{\rule{125pt}{0pt}}}$}
  \pr{125}{82}{$\mathcal{\bar{T}}$}
  \pr{2}{50}{$A$}
  \pr{2}{10}{$E$}
  \pr{125}{50}{$A$}
  \pr{248}{50}{$B$}
  \pr{248}{10}{$E$}
  \pr{-20}{30}{$\rho_{A E}$}
 \pr{278}{30}{$\rho_B \otimes \rho_E$}
\end{picture}
\end{center}

\caption{{\bf Decoupling.}  \label{fig:decoupling} The initial
  system, $A$, may be correlated to a reference system $E$. The
  evolution is modeled as a mapping $\mathcal{\bar{T}}$ from $A$ to
  $B$. The final state of $B$ is supposed to be independent of
  $E$. The subdivision of $\mathcal{\bar{T}}$ into a unitary
  $\mathcal{U}$ and a mapping $\mathcal{T}$ is required for the
  formulation of our decoupling criterion. }
\end{figure}

The decoupling criterion can be conceptually split into two parts, called achievability and converse part, which we now describe informally. The full technical statements are provided as
Theorems~\ref{thm:achievability} and~\ref{thm:converse} in Sections~\ref{sec:achievability} and~\ref{sec:converse}, respectively. For their formulation, it is convenient to view $\mathcal{\bar{T}}$ as a sequence, $\mathcal{\bar{T}} =\mathcal{T} \circ \mathcal{U}$, where $\mathcal{U}$ is an arbitrary unitary on $A$, and $\mathcal{T}$ a fixed TPCPM from $A$ to $B$.

\emph{Achievability: decoupling up to an error $\eps$ is achieved for most choices of
$\mathcal{U}$ if}
\begin{align}\label{eq:achiev}
H_{\min}^\eps(A|E)_{\rho}+H_{\min}^\eps(A|B)_{\tau}\gtrapprox 0\ .
\end{align}
\emph{Converse: decoupling up to an error $\eps$ is not achieved for any choice of $\mathcal{U}$ if}
\begin{align}
H_{\min}^\eps(A|E)_{\rho}+H_{\max}^\eps(A|B)_{\tau}\lessapprox 0\ .
\end{align}
The criteria refer to the $\eps$-smooth conditional min- and max-entropy introduced in~\cite{RenWol04b,renner-phd}, which can be seen as generalizations of the von Neumann entropy (cf.~Section~\ref{sec:preliminaries} for definitions and properties). The $\eps$-smooth conditional min-entropy $H_{\min}^\eps(A|E)_{\rho}$ is a measure for the correlation present in the initial state $\rho_{A E}$ | the larger this measure, the less dependent is $A$ on $E$ (see Table~\ref{tb:state} for some typical examples). The quantities $H_{\min}^{\eps}(A|B)_{\tau}$ (for the achievability) and $H_{\max}^{\eps}(A|B)_{\tau}$ (for the converse) measure how well the mapping $\mathcal{T}$ conserves correlations. Roughly, they quantify the uncertainty one has about a ``copy'' of the input, $A$, given access to the output, $B$, of $\mathcal{T}$ (cf.~Table~\ref{tb:mapping}). We note that the expressions for the achievability and for the converse essentially coincide in many cases of interest (see the discussion in Section~\ref{sec:converse}).

As a typical example for decoupling, consider $m$ qubits, $A$, that are classically maximally correlated to $E$ (so that $H_{\min}^{\eps}(A|E)_{\rho}=0$, cf.~second row of Table~\ref{tb:state}). Furthermore, assume that $A$ undergoes a reversible evolution, $\mathcal{U}$, after which we discard $m-m'$ qubits, corresponding to a partial trace, $\mathcal{T}= \tr_{m-m'}$ (see last example of Table~\ref{tb:mapping}). Our criterion then says that the remaining $m'$ qubits will, for most evolutions $\mathcal{U}$, be decoupled from $E$ whenever $m' <m/2$. Conversely, if this condition is not satisfied, some correlation will necessarily be retained.

\newcommand*{\theighta}{\vphantom{$ \displaystyle \sum_{i}^{2^k}$}}

\begin{table}[ht]
\begin{center}
{\small
\begin{tabular}{|p{6.2cm}|c|c|}
  \hline  Description of initial state & \vphantom{$\displaystyle \sum$} $\rho = \rho_{A E}$ &
  $H_{\min}^{\eps}(A|E)_{\rho}$ \\ \hline \hline
 $k$ random bits $A$ independent of $E$ & \theighta $2^{-k}\cdot\ident_{A} \otimes \rho_E$
  & $k$ \\ \hline
  $k$ bits $A$ correlated classically to $E$ & \theighta $\displaystyle 2^{-k}\cdot\sum_{i=1}^{2^k}
  \ketbra{i}_A \otimes \ketbra{i}_E$ & $0$ \\ \hline
  $k$ qubits $A$ fully entangled with $E$ & \theighta
  $\ketbra{\Psi}$, where $\displaystyle \Psi =
  2^{-k/2}\cdot\sum_{i=1}^{2^k} \ket{i}_A \otimes \ket{i}_E$ & $-k$ \\ \hline
\end{tabular}
}
\end{center}
\caption{{\bf Dependence on the initial state.}\label{tb:state}
The table illustrates how the term $H_{\min}^{\eps}(A|E)_{\rho}$ (for $\eps \to 0$) in the decoupling criterion depends on the initial state $\rho_{A E}$. In all three examples, $A$ is assumed to be a $k$-qubit system with orthonormal basis $\{\ket{i}_A\}_{i = 1}^{2^k}$.  Similarly, $\{\ket{i}_E\}_{i=1}^{2^k}$ is an orthonormal family of states on $E$.}

\end{table}

\begin{table}
\begin{center}
{\small
\begin{tabular}{|p{6.2cm}|c|c|}
  \hline  Description of mapping  & \vphantom{$\displaystyle \sum$}  $\mathcal{T}$ &
  $H_{\min}^{\eps}(A|B)_{\tau}$ \\ \hline \hline
  identity on $m$ qubits & \theighta $\sigma \mapsto \sigma$
  & $-m$ \\ \hline
  orthogonal measurement on $m$ qubits & \theighta $\displaystyle \sigma \mapsto
  \sum_{i=1}^{2^m} \ketbra{i} \sigma \ketbra{i}$ & $0$ \\ \hline
  erasure of $m$ qubits & \theighta
  $\sigma \mapsto \tr(\sigma) \ketbra{0} $ & $m$ \\ \hline 
 \raisebox{0.4ex}{\parbox{5.6cm}{identity on $m'$, orthogonal measurement on ${m-m'}$ qubits}}
  & \theighta
  $\displaystyle \sigma \mapsto \sum_{i=1}^{2^{m-m'}} (\ident_{m'} \otimes
  \ketbra{i}) \sigma (\ident_{m'} \otimes \ketbra{i})$ & $ -m'$ \\ \hline
  identity on $m'$, erasure on ${m-m'}$ qubits & \theighta
  $\sigma \mapsto \tr_{m-m'}(\sigma) $ & $m - 2m'$ \\ \hline
\end{tabular}
}
\end{center}

\caption{{\bf Dependence on the mapping.}\label{tb:mapping}
The table illustrates how the term $H_{\min}^\eps(A|B)_{\tau}$ in the decoupling criterion depends on the mapping $\mathcal{T}$. In all five examples, the input space, $A$, is assumed to consist of $m$  qubits  with orthonormal basis $\{\ket{i}_A\}_{i = 1}^{2^m}$. The last two examples have a smaller output space consisting of only $m'$ qubits. The penultimate one can be seen as a combination of the first and the second, and the last one can be seen as a combination of the first and the third. (The smooth conditional min-entropies are evaluated for $\eps \to 0$.)}
\end{table}

We mention that it is possible phrase our achievability criterion for decoupling~\eqref{eq:achiev} in another (but equivalent) way. For TPCPMs $\cT$ from $A$ to $B$ such that for every unitary $\cU$ on $A$ there exists a unitary $\cV$ on $B$ with $\cV\circ\cT=\cT\circ\cU$, decoupling up to an error $\eps$ is achieved if
\begin{align}
H_{\min}^\eps(A|E)_{\rho}+H_{\min}^\eps(A|B)_{\tau}\gtrapprox 0\ .
\end{align}
For more details about this formulation, see the discussion in Section~\ref{sec:achiev-statement}.

\paragraph{Applications.}

The notion of decoupling has various applications in information
theory and in physics. Many of these applications have in common that
decoupling of a system $B$ from a system $E$ is used to show that $B$
is maximally entangled with a complementary system, $R$. Indeed, under
the assumption that $R$ is chosen such that the joint state, $\rho_{B
  E R}$, is pure, $\rho_{B E} = \rho_{B} \otimes \rho_{E}$ immediately
implies that there exists a subsystem $R'$ of $R$ such that the state
on $\rho_{B R'}$ is pure. If, in addition, $\rho_B$ is fully mixed,
$\rho_{B R'}$ is necessarily maximally entangled.

In the context of information theory, this type of argument is, for example, used to analyze state merging~\cite{HoOpWi05Nat,HoOpWi07CMP}, i.e., the task of conveying a subsystem from a sender to a receiver | who already holds a possibly correlated subsystem | using classical communication and entanglement. Another example, where decoupling is used in a similar fashion, is the quantum reverse Shannon theorem~\cite{BSST02,BDHSW09,BeChRe09_2}. In fact, the proof of this theorem given in~\cite{BeChRe09_2} refers to a coherent form of state merging (also known as the fully quantum Slepian Wolf or mother protocol~\cite{FQSW}) where the classical communication is replaced by quantum communication. Decoupling can also be used for the characterization of correlation and entanglement between systems, erasure processes, as well as channel capacities (see, e.g.,
\cite{GrPoWi05,Buscemi09,HHWY08}). In addition, its classical analogue, privacy amplification~\cite{BBCM95,RenKoe05}, is widely used in classical and quantum cryptography.

Decoupling processes are also crucial in physics. For example, the evolution of a thermodynamical system towards thermal equilibrium can be understood as a decoupling process, where the system under consideration decouples from the observer (somewhat analogous to the considerations in~\cite{LPSW09,Partovi1,Partovi2}). Recent work indeed shows that there is a close relation between smooth entropies and quantities that are relevant in thermodynamics~\cite{DRRV09,RARDV10,diploma-hutter,Faist12,Aberg13,Horodecki13}. Similarly, black hole radiation may be analyzed from such a point of view~\cite{HayPre07,braunstein-pati,braunstein-zyczkowski}. Finally, one-shot decoupling techniques were also applied in solid state physics in order to show that 1D quantum states with exponential decay of correlations have an efficient classical approximate description as a matrix product state~\cite{Brandao12}.

\paragraph{History and Related Work.}

While various standard results in quantum information theory have been proved using ideas related to decoupling, the concept came into its own with the discovery of state merging protocols~\cite{HoOpWi05Nat,HoOpWi07CMP} and, later, the fully quantum Slepian Wolf protocol~\cite{FQSW}. These are based on specific decoupling processes where the mapping $\mathcal{T}$ is either a projective measurement or a partial trace. In this early work, the decoupling was analyzed in terms of the dimensions of certain subsystems (rather than smooth conditional entropies).

Based on the diploma thesis of one of us~\cite{diploma-berta}, we have generalized these decoupling results to include mappings~$\mathcal{T}$ that consist of combinations of projective measurements and partial trace-preserving. Furthermore, we expressed the decoupling criterion in terms of smooth conditional entropies. Subsequently, one of the authors derived in his doctoral thesis~\cite{fred-these} a general decoupling theorem that can be applied to any type of mapping. This result is essentially (up to the use of different entropy measures) equivalent to Theorem~\ref{thm:achievability} presented here. We also note that the aforementioned characterizations of decoupling can be seen as special cases of this general result.

The above work was mostly concerned with achievability. Converse results were so far only known in special cases. In particular, we derived in~\cite{BeReWi07} and~\cite{diploma-berta} (see also~\cite{Renner09}) converse theorems for the case where the mapping $\cT$ is a projective measurement. The converse theorem presented here, Theorem~\ref{thm:converse}, generalizes these results.

We emphasize that the use of smooth conditional entropies is essential for applications of the decoupling technique in physics (see the discussion in Section~\ref{sec:discussion}).

\paragraph{Structure of the Paper.}

In Section~\ref{sec:preliminaries} we introduce the notation and review the definitions and main properties of the entropy measures used in this work. Our main achievability result for decoupling is given in Section~\ref{sec:achievability}, whereas Section~\ref{sec:converse} contains a converse that is tight in many cases of interest.  The use of the decoupling technique is illustrated in Section~\ref{sec:applications}, where we show how to obtain optimal one-shot quantum state merging.  We conclude with a discussion in Section~\ref{sec:discussion}.

%%%%%%%%%%%%%%%%%%%%%%%%%%%%%%%%%%%%%%%%%%%%%%%%%%%%%%%%%%%%%%%%%%%%%

\section{Preliminaries} \label{sec:preliminaries}

\subsection{Notation}

We denote the Hilbert space associated to a system $A$ by $\cH_A$. We
only consider finite-dimensional systems and denote the dimension of
$\cH_{A}$ by $|A|$. The set of linear operators on $\cH$ is denoted by
$\cL(\cH)$ and the set of nonnegative operators on $\cH$ by
$\cP(\cH)$. We define the sets of subnormalized states
$\cS_{\leq}(\cH)=\{\rho\in\cP(\cH):\tr\rho\leq1\}$ and normalized
states $\cS_{=}(\cH)=\{\rho\in\cP(\cH):\tr\rho=1\}$.

The tensor product of $\cH_{A}$ and $\cH_{B}$ is denoted by
$\cH_{AB}=\cH_{A}\otimes\cH_{B}$. For multipartite operators
$\rho_{AB}\in\cP(\cH_{AB})$, we write $\rho_{A}=\tr_{B}(\rho_{AB})$
for the corresponding reduced operator. For $M_{A}\in\cL(\cH_{A})$, we
write $M_{A}=M_{A}\otimes\ident_{B}$ for the enlargement on any
$\cH_{AB}$, where $\ident_B$ denotes the identity in $\cP(\cH_B)$.

Completely positive maps from $\cL(\cH_{A})$ to $\cL(\cH_{B})$ are
called CPMs and trace-preserving CPMs are called TPCPMs. For
$\cH_{A}$, $\cH_{B}$ with orthonormal bases
$\{\ket{i}_{A}\}_{i=1}^{|A|}$, $\{\ket{i}_{B}\}_{i=1}^{|B|}$ and
$|A|=|B|$, the canonical identity mapping from $\cL(\cH_{A})$ to
$\cL(\cH_{B})$ with respect to these bases is denoted by
$\opid_{A\rightarrow B}$, i.e., $\opid_{A \rightarrow B}(\ket{i}\bra{j}_A) = \ket{i}\bra{j}_B$.

For $\rho\in\cP(\cH)$, $\|\rho\|_{\infty}$ denotes the operator norm of $\rho$, which is equal to the maximum eigenvalue of $\rho$.  The trace norm of $\rho\in\cL(\cH)$ is defined as $\|\rho\|_{1}=\tr(\sqrt{\rho^{\dagger}\rho})$ and the induced metric on $\cS_{\leq}(\cH)$ is called trace distance.\footnote{The trace distance is often defined with an additional factor $1/2$, which we omit here.} The fidelity between $\rho,\sigma\in\cS_{\leq}(\cH)$ is defined as $F(\rho,\sigma)=\|\sqrt{\rho}\sqrt{\sigma}\|_{1}$.

We will make use of the Choi-Jamio\l{}kowski isomorphism, which relates CPMs to positive operators, and which we denote by $J$.

\begin{lem}\cite{cj-jamiolkowski, cj-choi}\label{choi}
The Choi-Jamio\l{}kowski map $J$ takes maps $\cT_{A\rightarrow B}:\cL(\cH_{A})\rightarrow\cL(\cH_{B})$ to operators $J(\cT_{A\rightarrow B})\in\cL(\cH_{A}\otimes\cH_{B})$. It is defined as
\begin{align}
J(\cT_{A\rightarrow B})=(\opid_{A}\otimes\cT_{A'\rightarrow B})(\ket{\Phi}\bra{\Phi}_{AA'})\ ,
\end{align}
where $\ket{\Phi}_{AA'}=|A|^{-\frac{1}{2}}\sum_{i}\ket{i}_{A}\otimes\ket{i}_{A'}$ and $\cH_{A'}\cong\cH_{A}$.\footnote{The Choi-Jamio\l{}kowski isomorphism is sometimes defined with an additional dimensional factor of $|A|$; we choose not to do this here.}  The map $J$ bijectively maps the set of CPMs from $\cH_{A}$ to $\cH_{B}$ to the set $\cP(\cH_{A}\otimes\cH_{B})$, and its inverse maps any $\gamma_{AB}\in\cP(\cH_{A}\otimes\cH_{B})$ to
\begin{align}
\cT_{A \rightarrow B}: \, M_A \, \mapsto \, |A|\cdot\tr[\gamma_{A B}M_A^T] \ ,
\end{align}
where $M_A^T$ denotes the transpose of $M_A$ with respect to the basis $\{\ket{i}_A\}_{i=1}^{|A|}$.
\end{lem}

%%%%%%%%%%%%%%%%%%%%%%%%%%%%%%%%%%%%%%%%%%%%%%%%%%%%%%%%%%%%%%%%%%%%%

\subsection{Smooth Entropies}

The smooth entropy formalism~\cite{renner-phd, RenWol04b} has been
introduced in (classical and quantum) information theory to study
general one-shot scenarios, in which nothing needs to be assumed about
the structure of the relevant probability distributions or quantum
states (e.g., those modeling noise processes in a communication
channel). The formalism therefore overcomes a limitation of the
established theory, where it is usually assumed that the relevant
processes can be modeled as asymptotic sequences of independent and identically distributed (iid) subprocesses.

In this section we provide the definitions of the underlying entropy
measures, called smooth min- and max entropy, and state some of their
basic properties. Further properties are summarized in
Appendix~\ref{app_smooth}. For a more detailed discussion of the
smooth entropy formalism we refer to~\cite{Tomamichel12,renner-phd,min-max-entropy, ToCoRe09, duality-min-max-entropy, datta-2008-2}.

Recall the following standard definitions. The von Neumann entropy of
$\rho\in\cS_{=}(\cH)$ is defined as\footnote{All logarithms are taken to base $2$.} $H(\rho)=-\tr(\rho\log\rho)$ and the conditional von
Neumann entropy of $A$ given $B$ for $\rho_{AB}\in\cS_{=}(\cH)$ is
defined as $H(A|B)_{\rho}=H(AB)_{\rho}-H(B)_{\rho}$.

\begin{defin}
Let $\rho_{AB}\in\cS_{\leq}(\cH_{AB})$. The conditional min-entropy of $A$ given $B$ is defined as
\begin{align}
H_{\min}(A|B)_{\rho}=\sup_{\sigma_{B}\in\cS_{=}(\cH_{B})}\sup\big\{\lambda\in\mathbb{R}:2^{-\lambda}\cdot\ident_{A}\otimes\sigma_{B}-\rho_{AB}\geq0\big\}\ .
\end{align}
The conditional max-entropy of $A$ given $B$ is defined as
\begin{align}
H_{\max}(A|B)_{\rho}=\sup_{\sigma_{B}\in\cS_{=}(\cH_{B})}\log F(\rho_{AB},\ident_{A}\otimes\sigma_{B})^{2}\ .
\end{align}
\end{defin}

In the special case where $B$ is trivial (i.e., one-dimensional), we write $H_{\min}(A)_{\rho}$ and $H_{\max}(A)_{\rho}$ instead of $H_{\min}(A|B)_{\rho}$ and $H_{\max}(A|B)_{\rho}$, respectively, and it can be shown that $H_{\min}(A)_{\rho}=-\log\|\rho_{A}\|_{\infty}$ as well as $H_{\max}(A)_{\rho}=2\log\tr\sqrt{\rho_{A}}$. Furthermore, for $\rho_{AB}\in\cS_{=}(\cH_{AB})$ the entropies can be ordered as~\cite[Lemma 2]{ToCoRe09}
\begin{align}
H_{\min}(A|B)_{\rho}\leq H(A|B)_{\rho}\leq H_{\max}(A|B)_{\rho}\ .
\end{align}

The smooth conditional min- and max-entropy are defined by extremizing the non-smooth versions over a set of nearby states, where nearby is quantified by the purified distance.

\begin{defin}
Let $\rho,\sigma\in\cS_{\leq}(\cH)$. The purified distance between $\rho$ and $\sigma$ is defined as
\begin{align}
P(\rho,\sigma)=\sqrt{1-\bar{F}(\rho,\sigma)^{2}}\ ,
\end{align}
where $\bar{F}(\rho,\sigma)=F(\rho,\sigma)+\sqrt{(1-\tr[\rho])(1-\tr[\sigma])}$ denotes the generalized fidelity.
\end{defin}

The purified distance is a metric on $\cS_{\leq}(\cH)$~\cite[Lemma 5]{duality-min-max-entropy}. As its name indicates, $P(\rho,\sigma)$ corresponds to the minimum trace distance between purifications of $\rho$ and $\sigma$. For more about the purified distance we refer to~\cite{duality-min-max-entropy}.

Henceforth $\rho,\sigma\in\cS_{\leq}(\cH)$ are called $\eps$-close if $P(\rho,\sigma)\leq\eps$ and this is denoted by $\rho\approx_{\eps}\sigma$. We use the purified distance to specify an $\eps$-ball around $\rho\in\cS_{\leq}(\cH)$,
\begin{align}
B^{\eps}(\rho)=\{\rho'\in\cS_{\leq}(\cH):\rho'\approx_{\eps}\rho\}\ .
\end{align}

\begin{defin}
Let $\eps\geq0$ and $\rho_{AB}\in\cS_{\leq}(\cH_{AB})$. The $\eps$-smooth conditional min-entropy of $A$ given $B$ is defined as
\begin{align}
H_{\min}^{\eps}(A|B)_{\rho}=\sup_{\hat{\rho}_{AB}\in\cB^{\eps}(\rho_{AB})}H_{\min}(A|B)_{\hat{\rho}}\ .
\end{align}
The $\eps$-smooth conditional max-entropy of $A$ given $B$ is defined as
\begin{align}
H_{\max}^{\eps}(A|B)_{\rho}=\inf_{\hat{\rho}_{AB}\in\cB^{\eps}(\rho_{AB})}H_{\max}(A|B)_{\hat{\rho}}\ .
\end{align}
\end{defin}

We mention that the optimization problems defining the smooth conditional min- and max-entropy can be formulated as semi-definite programs~\cite[Section 5.2.1]{Tomamichel12}. This allows to efficiently compute them numerically.

The smooth conditional min- and max-entropy are dual to each other in the following sense.

\begin{lem}\cite[Lemma 16]{duality-min-max-entropy}\label{lem:dual}
Let $\eps\geq0$, $\rho_{AB}\in\cS_{\leq}(\cH_{AB})$ and let $\rho_{ABC}\in\cS_{\leq}(\cH_{ABC})$ be an arbitrary purification of $\rho_{AB}$. Then, we have that
\begin{align}
H_{\min}^{\eps}(A|B)_{\rho}=-H_{\max}^{\eps}(A|C)_{\rho}\ .
\end{align}
\end{lem}

Smooth entropies satisfy various natural properties analogous to those known for the von Neumann entropy. One of them is the invariance under local isometries.

\begin{lem}\cite[Lemma 13/15]{duality-min-max-entropy}\label{lem:isometry}
Let $\eps\geq0$, $\rho_{AB}\in\cS_{\leq}(\cH_{AB})$, and let $\cU_{A\rightarrow C}$ and $\cV_{B\rightarrow D}$ be isometries from $A$ to $C$ and $B$ to $D$, respectively. Then, we have that
\begin{align}
&H_{\min}^{\eps}(A|B)_{\rho}= H_{\min}^{\eps}(C|D)_{\cV\circ\cU(\rho)}\\
&H_{\max}^{\eps}(A|B)_{\rho}= H_{\max}^{\eps}(C|D)_{\cV\circ\cU(\rho)}\ .
\end{align}
\end{lem}

Another important property is the data processing inequality.

\begin{lem}\cite[Theorem 18]{duality-min-max-entropy}\label{lem:strong}
Let $\eps\geq0$, $\rho_{AB}\in\cS_{\leq}(\cH_{AB})$, and let $\cT_{B\rightarrow C}$ be a TPCPM from $B$ to $C$. Then, we have that
\begin{align}
&H_{\min}^{\eps}(A|B)_{\rho}\leq H_{\min}^{\eps}(A|C)_{\cT(\rho)}\\
&H_{\max}^{\eps}(A|B)_{\rho}\leq H_{\max}^{\eps}(A|C)_{\cT(\rho)}\ .
\end{align}
\end{lem}

Smooth entropies are generalizations of the von Neumann entropy, in the sense that the von Neumann entropy can be retrieved as a special case via the quantum asymptotic equipartition property (AEP).

\begin{lem}\cite[Corollary~6.6 and~6.7]{Tomamichel12}\label{thm:fully-quantum-aep}
Let $0 < \eps < 1$ and $\rho_{AB}\in\cS_{=}(\cH_{AB})$. Then, we have that
\begin{align}
& \lim_{n\rightarrow\infty}\frac{1}{n}H_{\min}^{\eps}(A|B)_{\rho^{\otimes n}}=H(A|B)_{\rho}\\
& \lim_{n\rightarrow\infty}\frac{1}{n}H_{\max}^{\eps}(A|B)_{\rho^{\otimes n}}=H(A|B)_{\rho}\ .
\end{align}
\end{lem}

For more properties of smooth entropies we refer to the Appendix~\ref{app_smooth} and~\cite{Tomamichel12,renner-phd, min-max-entropy,ToCoRe09,duality-min-max-entropy,datta-2008-2}.

For technical reasons we will also need the following auxiliary quantities.

\begin{defin}\label{h2}
Let $\rho_{AB}\in\cS_{\leq}(\cH_{AB})$. The conditional collision entropy of $A$ given $B$ is defined as
\begin{align}
H_{2}(A|B)_{\rho}=\sup_{\sigma_{B}\in\cS_{=}(\cH_{B})}-\log\tr\left[\left((\ident_{A}\otimes\sigma_{B}^{-1/4})\rho_{AB}(\ident_{A}\otimes\sigma_{B}^{-1/4})\right)^{2}\right]\ .
\end{align}
\end{defin}

\begin{defin}\label{hmax-rho-given-sigma}
Let $\rho_{AB} \in \cS_{\leqslant}(\cH_{AB})$ and $\sigma_B \in \cS_{\leqslant}(\cH_B)$. We define
\begin{align}
\hmax(A|B)_{\rho|\sigma}=\log F(\rho_{AB}, \ident_A \otimes \sigma_B)^2\ .
\end{align}
\end{defin}

It can be shown that $\hmax(A|B)_{\rho}=\sup_{\sigma \in \cS_{\leqslant}(\cH_B)} \hmax(A|B)_{\rho|\sigma}$.

\begin{defin}\label{hmin-rho-given-sigma}
Let $\rho_{AB}\in\cS_{\leqslant}(\cH_{AB})$ and $\sigma_B \in \cS_{\leqslant}(\cH_B)$. We define
\begin{align}
\hmin(A|B)_{\rho|\sigma}=\sup\big\{\lambda\in\mathbb{R}:2^{-\lambda}\cdot\ident_{A}\otimes\sigma_{B}-\rho_{AB}\geq0\big\}\ .
\end{align}
\end{defin}

It can be shown that  $\hmin(A|B)_{\rho} = \sup_{\sigma \in \cS_{\leqslant}(\cH_B)} \hmin(A|B)_{\rho|\sigma}$.

Finally, we note that, since all Hilbert spaces in this paper are assumed to have finite dimension, the infima and suprema in the expressions above can be replaced by minima and maxima, respectively.

%%%%%%%%%%%%%%%%%%%%%%%%%%%%%%%%%%%%%%%%%%%%%%%%%%%%%%%%%%%%%%%%%%%%%

\section{Achievability} \label{sec:achievability}

In this section, we present and prove a general decoupling theorem (Theorem~\ref{thm:achievability}), which corresponds to the achievability part of the criterion sketched informally in
Section~\ref{sec:Intro}. The theorem subsumes and extends previous results in this direction.

%%%%%%%%%%%%%%%%%%%%%%%%%%%%%%%%%%%%%%%%%%%%%%%%%%%%%%%%%%%%%%%%%%%%%

\subsection{Statement of the Decoupling Theorem}\label{sec:achiev-statement}

As explained in the introductory section (see Fig.~\ref{fig:decoupling}), we consider a mapping from a system $A$ to a system $B$. The mapping consists of a unitary on $A$, selected randomly according to the Haar measure over the unitary group on $\cH_A$, followed by an arbitrary mapping $\cT = \cT_{A \rightarrow B}$. In applications, $\cT$ often consists of a measurement or a partial trace (see Table~\ref{tb:mapping} for examples). The decoupling theorem then tells us how well the output, $B$, of the mapping $\cT$ is decoupled (on average over the choices of the unitary) from a reference system $E$. 

\begin{thm}[Decoupling Theorem]\label{thm:achievability}
Let $\eps>0$, $\rho_{AE}\in\cS_{=}(\cH_{AE})$, and let $\cT_{A \rightarrow B}$ be a CPM with Choi-Jamio\l{}kowski representation $\tau_{AB}=J(\cT)$ such that $\tr(\tau_{AB})\leqslant1$. Then, we have that
\begin{align}\label{eq:decouplbound}
\int_{\mbU(A)}\left\| \cT_{A\rightarrow B}(U_{A}\rho_{AE}U_{A}\mdag)-\tau_B\otimes\rho_E\right\|_1 dU\leqslant 2^{-\demi H_{\min}^{\eps}(A|E)_{\rho}-\demi H_{\min}^{\eps}(A|B)_{\tau}}+12\eps\ ,
\end{align}
where $\int\cdot\,dU$ denotes the integral over the Haar measure over the full unitary group on $\cH_A$.
\end{thm}

Here, the total CPM is of the form $\bar{\cT}=\cT\circ\cU$ with the unitary channel $\cU(\cdot)=U_{A}(\cdot)U_{A}^\dagger$ and $U_{A}$ chosen at random. We note that, equivalently, we may think of $\bar{\cT}$ as a channel that chooses at random a unitary $U_{A}$ and outputs the choice of $U_{A}$, together with the output of $\cT$.

The decoupling theorem (Theorem~\ref{thm:achievability}) provides a bound on the quality of decoupling that only depends on two entropic quantities, $H_{\min}^{\eps}(A|E)_{\rho}$ and $H_{\min}^{\eps}(A|B)_{\tau}$. The first is a measure for the correlations between $A$ and $E$ that are present in the initial state, $\rho_{AE}$. The second quantifies properties of the mapping $\cT$, which is characterized by the bipartite state $\tau_{A B}$ obtained via the Choi-Jamio\l{}kowski isomorphism $J$. Hence, in order to minimize the right hand side of~\eqref{eq:decouplbound}, no channel ends up being better suited for some types of states than for others or vice-versa. Furthermore, as discussed in Section~\ref{sec:converse}, the bound in~\eqref{eq:decouplbound} is essentially optimal in many cases of interest. We also note that, using Markov's inequality, the expectation value over the unitaries $U$ can be turned into a bound that holds for most unitaries. That is, for any $\mu > 0$,
\begin{align}
\left\| \cT_{A\rightarrow B}(U_{A}\rho_{AE}U_{A}\mdag)-\tau_B \otimes \rho_E \right\|_1 \leqslant \frac{1}{\mu}\cdot2^{-\demi H_{\min}^{\eps}(A|E)_{\rho}-\demi H_{\min}^{\eps}(A|B)_{\tau}}+\frac{12\eps}{\mu}
\end{align}
holds with probability at least $1-\mu$ (for $U$ chosen according to the Haar measure).

Finally, as sketched in the introductory section, the decoupling theorem (Theorem~\ref{thm:achievability}) can also be phrased in another (but equivalent) way.

\begin{cor}\label{cor:alt-achiev}
Let $\eps>0$, $\rho_{AE}\in\cS_{=}(\cH_{AE})$, and let $\cT_{A \rightarrow B}$ be a CPM with Choi-Jamio\l{}kowski representation $\tau_{AB}=J(\cT)$ such that $\tr(\tau_{AB})\leqslant1$. Furthermore, assume that for every unitary channel $\cU_{A}$ there exists a unitary channel $\cV_{B}$ such that $\cV_{B}\circ\cT_{A\rightarrow B}=\cT_{A\rightarrow B}\circ\cU_{A}$. Then, we have that
\begin{align}
\left\|\cT_{A\rightarrow B}(\rho_{AE})-\tau_B\otimes\rho_E \right\|_1\leqslant 2^{-\demi H_{\min}^{\eps}(A|E)_{\rho}-\demi H_{\min}^{\eps}(A|B)_{\tau}}+12\eps\ .
\end{align}
\end{cor}

\begin{proof}
By the decoupling theorem (Theorem~\ref{thm:achievability}) for the map $\cT_{A \rightarrow B}$, there exists a unitary $U_{A}$ such that
\begin{align}
\left\| \cT_{A\rightarrow B}(U_{A} \rho_{AE} U_{A}\mdag)-\tau_B \otimes \rho_E \right\|_1\leqslant 2^{-\demi H_{\min}^{\eps}(A|E)_{\rho} -\demi H_{\min}^{\eps}(A|B)_{\tau}}+12\eps\ .
\end{align}
Since there exists by assumption a unitary $V_{B}$ such that $\cV_{B}\circ\cT_{A\rightarrow B}=\cT_{A\rightarrow B}\circ\cU_{A}$, we get
\begin{align}\label{eq:alternative}
\left\|\cT_{A\rightarrow B}(\rho_{AE})-V_{B}\mdag\tau_B V_{B}\otimes\rho_E\right\|_1&=\left\|V_{B}\cT_{A\rightarrow B}(\rho_{AE})V_{B}\mdag-\tau_B\otimes\rho_E\right\|_1\nonumber\\
&\leqslant2^{-\demi H_{\min}^{\eps}(A|E)_{\rho} -\demi H_{\min}^{\eps}(A|B)_{\tau}}+12\eps\ .
\end{align}
Furthermore, again by assumption, there exists a unitary $W_{B}$ such that $\cW_{B}\circ\cT_{A\rightarrow B}=\cT_{A\rightarrow B}\circ\cU_{A}\mdag$, and hence
\begin{align}
\cT_{A\rightarrow B}=\cT_{A\rightarrow B}\circ\cU_{A}\circ\cU_{A}\mdag=\cV_{B}\circ\cT_{A\rightarrow B}\circ\cU_{A}\mdag=\cV_{B}\circ\cW_{B}\circ\cT_{A\rightarrow B}\ .
\end{align}
This implies $\cV_{B}\mdag\circ\cT_{A\rightarrow B}=\cW_{B}\circ\cT_{A\rightarrow B}$, and thus we get
\begin{align}
V_{B}\mdag\tau_B V_{B}=W_{B}\cT_{A\rightarrow B}\left(\frac{\ident_{A}}{|A|}\right)W_{B}\mdag=\cT_{A\rightarrow B}\left(U_{A}\mdag\frac{\ident_{A}}{|A|}U_{A}\right)=\cT_{A\rightarrow B}\left(\frac{\ident_{A}}{|A|}\right)=\tau_{B}\ .
\end{align}
Finally, we arrive at the claim by combining this with~\eqref{eq:alternative} .
\end{proof}

To see why this alternative formulation (Corollary~\ref{cor:alt-achiev}) is equivalent to the decoupling theorem (Theorem~\ref{thm:achievability}) we may think of the total map in Theorem~\ref{thm:achievability} as a channel that chooses at random a unitary $U_{A}$ and outputs the choice of $U_{A}$, together with the output of $\cT$. By inspection, this total map then fulfills the assumption of Corollary~\ref{cor:alt-achiev}.

Our first step in proving Theorem~\ref{thm:achievability} is to prove a version involving non-smooth min-entropies (Theorem~\ref{thm:nonsmooth-mainthm}). Then, in a second step, we show that smoothing preserves the essence of the theorem. Note that Theorem~\ref{thm:nonsmooth-mainthm} may be of interest in cases where no smoothing is required since it is slightly more general: it applies to any completely positive $\cT$, not only trace-non-increasing ones.

\begin{thm}[Non-Smooth Decoupling Theorem]\label{thm:nonsmooth-mainthm}
Let $\rho_{AE}\in\cS_{\leq}(\cH_{AE})$ and let $\cT_{A \rightarrow B}$ be a CPM with Choi-Jamio\l{}kowski representation $\tau_{AB}=J(\cT)$. Then, we have that
\begin{align}
\int_{\mbU(A)} \left\| \cT_{A\rightarrow B}(U_{A}\rho_{AE}U_{A}\mdag)-\tau_B\otimes\rho_E\right\|_1 dU\leqslant 2^{-\demi H_2(A|E)_{\rho}-\demi H_2(A|B)_{\tau}}\ ,
\end{align}
where $\int\cdot\,dU$ denotes the integral over the Haar measure over the full unitary group on $\cH_A$.
\end{thm}

%%%%%%%%%%%%%%%%%%%%%%%%%%%%%%%%%%%%%%%%%%%%%%%%%%%%%%%%%%%%%%%%%%%%%

\subsection{Technical Ingredients to the Proof}

The proof of the non-smooth decoupling theorem (Theorem~\ref{thm:nonsmooth-mainthm}) is based on a few technical lemmas, which we state and prove in the following, and which may be of independent interest. We note that they partly generalize techniques developed in the context of privacy amplification~\cite{RenKoe05,renner-phd,leftover} as well as earlier work on decoupling (see, e.g., \cite{HoOpWi07CMP}).

\begin{lem}\label{lem:swap-trick}
Let $M,N\in\cL(\cH_{A})$. Then, we have that $\tr[(M \otimes N) F]=\tr[MN]$, where $F$ swaps the two copies of the $A$ subsystem.
\end{lem}

\begin{proof}
Write $M$ and $N$ in the standard basis for $\cH_{A}$, that is, $M = \sum_{ij} m_{ij} \ket{i}\bra{j}$ and $N = \sum_{kl} n_{kl} \ket{k}\bra{l}$. Then, we have that
\begin{align}
\tr[(M \otimes N) F] &= \tr\left[ \left( \sum_{ijkl} m_{ij} n_{kl} \ket{i}\bra{j} \otimes \ket{k}\bra{l} \right)F \right]\nonumber\\
&=\tr\left[ \sum_{ijkl} m_{ij} n_{kl} \ket{i}\bra{l} \otimes \ket{k}\bra{j} \right]\nonumber\\
&=\sum_{ij} m_{ij} n_{ji}\nonumber\\
&=\tr[MN]\ .
\end{align}
\end{proof}

The second lemma involves averaging over Haar distributed unitaries. While it would take us too far afield to formally introduce the Haar measure, it can simply be thought of as the uniform probability distribution over the set of all unitaries on a Hilbert space. The following then tells us the expected value of $U^{\otimes 2}M(U^{\dagger})^{\otimes 2}$ with $M\in\cL(\cH_{A}^{\otimes2})$ when $U$ is selected ``uniformly at random''.

\begin{lem}\label{lem:haar-integral}
Let $M\in\cL(\cH_{A}^{\otimes2})$. Then, we have that
\begin{align}
\mbE(M) = \int_{\mbU(A)} U^{\otimes 2}M(U^{\dagger})^{\otimes2}dU=\alpha\cdot\ident_{A A'}+\beta\cdot F_{A}\ ,
\end{align}
where $F_{A}$ swaps the two copies of the $A$ subsystem, $\alpha$ and $\beta$ are such that $\tr[M] = \alpha |A|^2 + \beta |A|$ and $\tr[MF] = \alpha |A| + \beta |A|^2$, and $dU$ is the normalized Haar measure on $\mbU(A)$.
\end{lem}

\begin{proof}
This follows directly from a standard result in Schur-Weyl duality, e.g., \cite[Proposition~2.2]{collins-sniady}. The latter states that $\mbE:\cL(\cH_{A}^{\otimes2})\rightarrow\cL(\cH_{A}^{\otimes2})$ is an orthogonal projection onto $\Span\{\ident,F\}$ under the inner product $\langle A,B \rangle = \tr[A\mdag B]$. Hence, $\mbE(M)$ can be written as $\alpha\cdot\ident_{AA'}+\beta\cdot F_{A}$ as claimed, and the conditions $\tr[\ident \mbE(M)] = \tr[M]$ and $\tr[F \mbE(M)] =\tr[FM]$ must be fulfilled, and these lead to the two conditions on $\alpha$ and $\beta$.
\end{proof}

The following bounds the ratio of the purity of a bipartite state and the purity of the reduced state on one subsystem.

\begin{lem}\label{lem:tr2-bounded}
Let $\xi_{AB} \in\cP(\cH_{AB})$. Then, we have that
\begin{align}
\frac{1}{|A|} \leqslant \frac{\tr\left[ {\xi_{AB}}^2 \right]}{\tr\left[ {\xi_B}^2 \right]} \leqslant |A|\ .
\end{align}
\end{lem}

\begin{proof}
Letting $A'$ be a system isomorphic to $A$, we first prove the left-hand side
\begin{align}
\tr\left[ {\xi_B}^2 \right] &= \tr\left[ \tr_A\left[ \xi_{AB} \right]^2 \right]\nonumber\\
&= \tr\left[ \tr_A\left[ \xi_{AB} \right]\cdot\tr_{A'}\left[ \xi_{A'B} \right] \right]\nonumber\\
&= \tr\left[ \xi_{AB} \left( \tr_{A'}\left[ \xi_{A'B}\right]  \otimes \ident_A \right) \right]\nonumber\\
&= \tr\left[ (\xi_{AB} \otimes \ident_{A'})(\xi_{A'B}  \otimes \ident_A ) \right]\nonumber\\
&\leqslant \sqrt{\tr \left[ (\xi_{AB} \otimes \ident_{A'})^2 \right]\cdot\tr\left[ (\xi_{A'B}\otimes\ident_{A})^2 \right]}\nonumber\\
&= \tr\left[ {\xi_{AB}}^2 \otimes \ident_{A'} \right]\nonumber\\
&=|A|\cdot\tr\left[ {\xi_{AB}}^2 \right]\ ,
\end{align}
where the inequality is due to an application of Cauchy-Schwarz. The right-hand side follows from the fact that $\xi_{AB} \leqslant |A|\cdot\ident_A \otimes \xi_B$. This can in turn be seen from the fact that we can write
\begin{align}
|A|\cdot\ident_A \otimes \xi_B=\sum_{i=1}^{|A|^{2}}U_{A}^{i}\xi_{AB}(U_{A}^{i})\mdag\ ,
\end{align}
with unitaries $U_{A}^{i}$ such that $\tr\left[(U_{A}^{i})\mdag U_{A}^{j}\right]=0$ for every $i\neq j$, and $U_{A}^{1}=\ident_{A}$.
\end{proof}

In the main proof, we will need to bound the trace distance between two states. The following lemma will allow us to do this.

\begin{lem}\label{lem:pseudo-jensen-renato}
Let $M\in\cL(\cH_{A})$ and $\sigma\in\cP(\cH_{A})$. Then, we have that
\begin{align}
\| M \|_1 \leqslant \sqrt{\tr[\sigma]\cdot\tr[\sigma^{-1/4} M \sigma^{-1/2} M\mdag \sigma^{-1/4}]}\ .
\end{align}
In particular, if $M$ is Hermitian then, we have that
\begin{align}
\| M \|_1 \leqslant \sqrt{\tr[\sigma]\cdot\tr[(\sigma^{-1/4} M \sigma^{-1/4})^2]}\ .
\end{align}
\end{lem}

This is a slight generalization of~\cite[Lemma 5.1.3]{renner-phd}. For completeness we give a different proof here.

\begin{proof}
We calculate
\begin{align}
\| M \|_1 &= \max_{U} \left| \tr[UM] \right|\nonumber\\
&= \max_U \left| \tr\big[(\sigma^{1/4} U \sigma^{1/4})(\sigma^{-1/4} M \sigma^{-1/4})\big] \right|\nonumber\\
&\leqslant \max_U \sqrt{\tr\left[(\sigma^{1/4} U \sigma^{1/4})(\sigma^{1/4} U\mdag \sigma^{1/4})\right]\cdot\tr\left[ \sigma^{-1/4} M \sigma^{-1/2}  M\mdag \sigma^{-1/4}\right]}\nonumber\\
&= \sqrt{\max_U \tr[\sigma^{1/2} U \sigma^{1/2} U\mdag]\cdot\tr\left[ \sigma^{-1/4} M \sigma^{-1/2} M\mdag \sigma^{-1/4} \right]}\nonumber\\
&= \sqrt{\tr[\sigma]\cdot\tr\left[ \sigma^{-1/4} M \sigma^{-1/2} M\mdag \sigma^{-1/4} \right]}\ ,
\end{align}
where the inequality results from an application of Cauchy-Schwarz, and the maximizations are over all unitaries on $A$. The last equality follows from
\begin{align}
\max_U \tr\left[\sigma^{1/2} U \sigma^{1/2} U\mdag\right] &\leqslant \max_U \sqrt{\tr\left[\sigma\right]\cdot\tr\left[U \sigma^{1/2} U\mdag U \sigma^{1/2} U\mdag\right]}\nonumber\\
&= \tr[\sigma]\nonumber\\
&\leqslant \max_U \tr[\sigma^{1/2} U \sigma^{1/2} U\mdag]\ .
\end{align}
\end{proof}

%%%%%%%%%%%%%%%%%%%%%%%%%%%%%%%%%%%%%%%%%%%%%%%%%%%%%%%%%%%%%%%%%%%%%

\subsection{Proof of the Non-Smooth Decoupling Theorem (Theorem~\ref{thm:nonsmooth-mainthm})}

Throughout the proof, we will denote with a prime the ``twin'' subsystems used when we take tensor copies of operators, and $F_{S}$ denotes a swap between $S$ and $S'$.

We first use Lemma \ref{lem:pseudo-jensen-renato} to bound the trace norm. For $\sigma_{B}\in\cS_{=}(\cH_{B})$ and $\zeta_{E}\in\cS_{=}(\cH_{E})$ we get
\begin{align}
&\|\cT_{A\rightarrow B}(U_{A}\rho_{AE}U_{A}\mdag)-\tau_{B}\otimes\rho_{E}\|_{1}\nonumber\\
&\leqslant \sqrt{\tr\left[ \left( (\sigma_B \otimes \zeta_E)^{-1/4} (\cT_{A\rightarrow B}(U_{A} \rho_{AE} U_{A}\mdag) - \tau_B \otimes \rho_E)(\sigma_B \otimes \zeta_E)^{-1/4}\right)^2 \right]}\ .
\end{align}
Now define the CPM $\tilde{\cT}_{A \rightarrow B}(\cdot)=\sigma_{B}^{-1/4}\cT_{A\rightarrow B}(\cdot)\sigma_{B}^{-1/4}$ and the operators $\tilde{\tau}_{A'B} =J(\tilde{\cT})$ and $\tilde{\rho}_{AE} = \zeta_{E}^{-1/4}\rho_{AE}\zeta_{E}^{-1/4}$. We then rewrite the above as
\begin{align}
\left\| \cT_{A\rightarrow B}(U_{A} \rho_{AE} U_{A}\mdag) - \tau_B \otimes \rho_E \right\|_1 \leqslant \sqrt{\tr\left[ \left( \mathcal{\tilde{T}}_{A\rightarrow B}(U_{A} \tilde{\rho}_{AE} U_{A}\mdag) - \tilde{\tau}_B \otimes \tilde{\rho}_E\right)^2 \right]}\ .
\end{align}
Using Jensen's inequality we obtain
\begin{align}\label{eqn:nonsmooth-mainthm-inter1}
\int \left\| \cT_{A\rightarrow B}(U_{A} \rho_{AE} U_{A}\mdag) - \tau_B \otimes \rho_E \right\|_1 dU \leqslant \sqrt{\int \tr\left[ \left( \mathcal{\tilde{T}}_{A\rightarrow B}(U_{A} \tilde{\rho}_{AE} U_{A}\mdag) - \tilde{\tau}_B \otimes \tilde{\rho}_E\right)^2 \right] dU}\ .
\end{align}
We now simplify the integral
\begin{align}
&\int\tr\left[ \left( \tilde{\cT}_{A\rightarrow B}(U_{A} \tilde{\rho}_{AE} U_{A}\mdag) - \tilde{\tau}_B \otimes \tilde{\rho}_E \right)^2 \right] dU\nonumber\\
&= \int \tr\left[ \left( \tilde{\cT}_{A\rightarrow B}(U_{A} \tilde{\rho}_{AE} U_{A}\mdag) \right)^2 \right]dU - 2 \int \tr\left[ \tilde{\cT}_{A\rightarrow B}(U_{A} \tilde{\rho}_{AE} U_{A}\mdag) \left( \tilde{\tau}_B \otimes \tilde{\rho}_E \right) \right]dU + \tr\left[ \left( \tilde{\tau}_B \otimes \tilde{\rho}_E \right)^2 \right]\nonumber\\
&= \int \tr\left[ \left( \tilde{\cT}_{A\rightarrow B}(U_{A} \tilde{\rho}_{AE} U_{A}\mdag) \right)^2 \right]dU - 2 \tr\left[ \tilde{\cT}_{A\rightarrow B}\left(\int U_{A} \tilde{\rho}_{AE} U_{A}\mdag dU\right)  \left( \tilde{\tau}_B \otimes\tilde{\rho}_E \right) \right] + \tr\left[ \left( \tilde{\tau}_B \otimes \tilde{\rho}_E \right)^2 \right]\nonumber\\
&= \int \tr\left[ \left( \tilde{\cT}_{A\rightarrow B}(U_{A} {\tilde{\rho}}_{AE} U_{A}\mdag) \right)^2 \right]dU - \tr\left[ \tilde{\tau}_{B}^{2} \right]\cdot\tr\left[ \tilde{\rho}_{E}^{2} \right]\ .
\end{align}
We rewrite the first term as follows
\begin{align}\label{eqn:attack-first-term}
\nonumber\int \tr\left[ \left( \tilde{\cT}_{A\rightarrow B}(U_{A} \tilde{\rho}_{AE} U_{A}\mdag) \right)^2 \right]dU &= \int \tr\left[ \left( \tilde{\cT}_{A\rightarrow B}(U_{A} \tilde{\rho}_{AE} U_{A}\mdag)  \right)^{\otimes 2} F_{BE}\right] dU\\
\nonumber&= \int \tr\left[ \left( \tilde{\cT}^{\otimes 2}_{A\rightarrow B}\left(U_{A}^{\otimes 2} \tilde{\rho}_{AE}^{\otimes 2} (U^{\dagger}_{A})^{\otimes 2}\right) \right) F_{BE}\right] dU\\
\nonumber&= \int \tr\left[ \tilde{\rho}_{AE}^{\otimes 2}  \left( \left( (U_{A}^{\dagger})^{\otimes 2} (\tilde{\cT}_{B\rightarrow A}\mdag)^{\otimes 2}(F_{B})U_{A}^{\otimes 2} \right) \otimes F_E \right) \right] dU\\
&=\tr\left[ \tilde{\rho}_{AE}^{\otimes 2}\left(\left(\int(U_{A}^{\dagger})^{\otimes 2} (\tilde{\cT}_{B\rightarrow A}\mdag)^{\otimes 2}(F_{B})U_{A}^{\otimes 2}dU\right)\otimes F_E  \right)\right]\ ,
\end{align}
where we have used the swap trick (Lemma \ref{lem:swap-trick}) with $F_{BE}=F_{B}\otimes F_{E}$ in the first equality, the definition of the adjoint of a superoperator in the third equality and the linearity of the trace in forth equality. We now compute the integral using a lemma about Haar distributed unitaries (Lemma \ref{lem:haar-integral})
\begin{align}
\int (U_{A}^{\dagger})^{\otimes 2} (\tilde{\cT}_{B\rightarrow A}\mdag)^{\otimes 2}(F_{B}) U_{A}^{\otimes 2}dU=\alpha\cdot\ident_{AA'}+\beta\cdot F_{A}\ ,
\end{align}
where $\alpha$ and $\beta$ satisfy the following equations
\begin{align}
\alpha|A|^2+\beta|A|=\tr\left[(\tilde{\cT}_{B\rightarrow A}\mdag)^{\otimes 2}(F_{B})\right]= \tr\left[F_B\tilde{\cT}_{A\rightarrow B}^{\otimes 2}(\ident_{AA'})\right]&=|A|^2\cdot\tr\left[ F_B \tilde{\tau}_B^{\otimes 2} \right]\nonumber\\
&=|A|^2\cdot\tr\left[\tilde{\tau}_{B}^{2}\right]
\end{align}
and
\begin{align}
\alpha|A|+\beta|A|^2&=\tr\left[(\tilde{\cT}_{B\rightarrow A}\mdag)^{\otimes 2}(F_{B})  F_A \right]\nonumber\\
&=\tr\left[ F_{B} \tilde{\cT}^{\otimes 2}_{A\rightarrow B}(F_{A})\right]\nonumber\\
&=|A|^2\cdot\tr\left[ F_B\cdot\tr_{AA'} \left[ \tilde{\tau}_{AB}^{\otimes 2} (F_A \otimes \ident_{BB'}) \right] \right]\nonumber\\
&=|A|^2\cdot\tr\left[ (\ident_{AA'} \otimes F_B) \tilde{\tau}_{AB}^{\otimes 2}(F_A \otimes \ident_{BB'}) \right]\nonumber\\
&=|A|^2\cdot\tr\left[ F_{AB} \tilde{\tau}_{AB}^{\otimes 2} \right]\nonumber\\
&=|A|^2\cdot\tr\left[ \tilde{\tau}_{AB}^2 \right]\ .
\end{align}
In the third equality, we have used the fact that $\tilde{\tau}_{AB}$ is a Choi-Jamio\l{}kowski representation of $\tilde{\cT}$ (Lemma~\ref{choi}), and the fourth equality is due to the fact that the adjoint of the partial trace is tensoring with the identity. Solving this system of equations yields
\begin{align}
\alpha &= \tr\left[ \tilde{\tau}_{B}^{2} \right]\cdot\left( \frac{|A|^2 - \frac{|A|\cdot\tr\left[ \tilde{\tau}_{AB}^{2} \right]}{\tr\left[ \tilde{\tau}_{B}^{2}\right]}}{|A|^2-1} \right)\\
\beta &= \tr\left[ \tilde{\tau}_{AB}^{2} \right]\cdot\left( \frac{|A|^2 - \frac{|A|\cdot\tr\left[ \tilde{\tau}_{B}^{2} \right]}{\tr\left[ \tilde{\tau}_{AB}^{2}\right]}}{|A|^2-1} \right)\ .
\end{align}
By applying Lemma \ref{lem:tr2-bounded}, we can simplify this to $\alpha \leqslant \tr\left[ \tilde{\tau}_{B}^{2} \right]$ and $\beta \leqslant \tr\left[ \tilde{\tau}_{AB}^{2} \right]$. Substituting this into (\ref{eqn:attack-first-term}) and using the swap trick twice (Lemma \ref{lem:swap-trick}), and then substituting into (\ref{eqn:nonsmooth-mainthm-inter1}) yields
\begin{align}
\int \left\| \cT_{A\rightarrow B}(U_{A} \rho_{AE} U_{A}\mdag) - \tau_B \otimes \rho_E \right\|_1 dU \leqslant \sqrt{\tr\left[ \tilde{\tau}_{AB}^{2} \right]\cdot\tr\left[ \tilde{\rho}_{AE}^{2} \right]}\ .
\end{align}
Finally we get the theorem by using the definitions of $\tilde{\tau}_{AB}$, $\tilde{\rho}_{AE}$ and the definition of the conditional collision entropy (Definition~\ref{h2}).
\qed

%%%%%%%%%%%%%%%%%%%%%%%%%%%%%%%%%%%%%%%%%%%%%%%%%%%%%%%%%%%%%%%%%%%%%

\subsection{Proof of the Main Decoupling Theorem (Theorem~\ref{thm:achievability})}

We now prove our main result, which is obtained from the non-smooth decoupling theorem (Theorem~\ref{thm:nonsmooth-mainthm}) by replacing the conditional collision entropies by smooth conditional min-entropies. 

First, note that the conditional collision entropy is always greater or equal to the conditional min-entropy (Lemma~\ref{lem:h2-hmin}) and therefore we are allowed to replace the $H_2$ terms on the right-hand side of the statement of Theorem~\ref{thm:nonsmooth-mainthm} by $H_{\min}$ terms. Thus we only have to consider the smoothing.

Let $\widehat{\rho}^{AE}\in\cB^{\eps}(\rho_{AE})$ be such that $H_{\min}^{\eps}(A|E)_{\rho}=H_{\min}(A|E)_{\widehat{\rho}}$ and $\widehat{\tau}_{AB}\in\cB^{\eps}(\tau_{AB})$ be such that $H_{\min}^{\eps}(A|B)_{\tau}=H_{\min}(A|B)_{\widehat{\tau}}$.

Furthermore write $\widehat{\tau}_{AB}-\tau_{AB}=\Delta^{+}_{AB}-\Delta^{-}_{AB}$, where $\Delta^{\pm}_{AB}\in\cP(\cH_{AB})$ have orthogonal support, and likewise, $\widehat{\rho}_{AE}-\rho_{AE}=\delta^{+}_{AE}-\delta^{-}_{AE}$ with $\delta^{+}_{AE}$ and $\delta^{-}_{AE}$ having orthogonal support as well as $\delta^{\pm}_{AE}\in\cP(\cH_{AE})$. By the equivalence of purified distance and trace distance (Lemma~\ref{lem:purified-dist-vs-trace-dist}) we have $\| \widehat{\tau}_{AB}-\tau_{AB}\|_1 \leqslant 2\eps$ and hence $\left\|\Delta^{\pm}_{AB}\right\|_1 \leqslant 2\eps$.

Moreover define $\widehat{\cT}_{A\rightarrow B}$, $\mathcal{D}^{-}_{A\rightarrow B}$ and $\mathcal{D}^{+}_{A\rightarrow B}$ as the unique superoperators that are such that $\widehat{\tau}_{AB}=J(\widehat{\cT}_{A\rightarrow B})$, $\Delta^{-}_{AB} = J(\mathcal{D}^{-}_{A\rightarrow B})$ and $\Delta^{+}_{AB} = J(\mathcal{D}^{+}_{A\rightarrow B})$, respectively.

Using the non-smooth decoupling theorem (Theorem~\ref{thm:nonsmooth-mainthm}) we get
\begin{align}
2^{-\demi H_{\min}^{\eps}(A|B)_{\tau}-\demi H_{\min}^{\eps}(A|E)_{\rho}}&\geqslant\int\left\|\widehat{\cT}_{A\rightarrow B}(U_{A}\widehat{\rho}_{AE}U_{A}\mdag) - \widehat{\tau}_B \otimes \widehat{\rho}_E\right\|_1 dU\nonumber\\
&\geqslant\int\left\|\widehat{\cT}_{A\rightarrow B}(U_{A}\widehat{\rho}_{AE}U_{A}\mdag) - \tau_B \otimes \rho_E \right\|_1 dU-4\eps\nonumber\\
&\geqslant\int\left\|\cT_{A\rightarrow B}(U_{A}\rho_{AE}U_{A}\mdag)-\tau_B\otimes\rho_E\right\|_1 dU\nonumber\\
&-\int\left\|\widehat{\cT}_{A\rightarrow B}(U_{A}\rho_{AE}U_{A}\mdag)-\widehat{\cT}_{A\rightarrow B}(U_{A}\widehat{\rho}_{AE}U_{A}\mdag)\right\|_1 dU\nonumber\\
&-\int\left\|\cT_{A\rightarrow B}(U_{A}\rho_{AE}U_{A}\mdag)-\widehat{\cT}_{A\rightarrow B}(U_{A}\rho_{AE}U_{A}\mdag)\right\|_1dU-4\eps\ ,
\end{align}
where we have used the triangle inequality for the trace distance in the second inequality. We now deal with the second term above
\begin{align}
&\int\left\|\widehat{\cT}_{A\rightarrow B}(U_{A}\rho_{AE}U_{A}\mdag)-\widehat{\cT}_{A\rightarrow B}(U_{A}\widehat{\rho}_{AE}U_{A}\mdag) \right\|_1 dU\nonumber\\
&= \int \left\| \widehat{\cT}_{A\rightarrow B}(U_{A}(\delta^{+}_{AE} - \delta^{-}_{AE}) U_{A}\mdag) \right\|_1 dU\nonumber\\
&\leqslant \int \left\| \widehat{\cT}_{A\rightarrow B}(U_{A}\delta^{+}_{AE} U_{A}\mdag) \right\|_1 dU + \int \left\| \widehat{\cT}_{A\rightarrow B}(U_{A}\delta^{-}_{AE} U_{A}\mdag) \right\|_1 dU\nonumber\\
&= \int \tr\left[\widehat{\cT}_{A\rightarrow B}(U_{A}\delta^{+}_{AE} U_{A}\mdag)\right]dU + \int\tr\left[\widehat{\cT}_{A\rightarrow B}(U_{A}\delta^{-}_{AE}U_{A}\mdag)\right]dU\nonumber\\
&= \tr\left[\widehat{\cT}_{A\rightarrow B}\left(\frac{\ident_{A}}{|A|}\right) \right]\cdot\left(\tr\left[\delta^{+}_{AE}\right] + \tr\left[\delta^{-}_{AE}\right]\right)\nonumber\\
&\leqslant 4\eps\ .
\end{align}
We deal with the third term in a similar fashion
\begin{align}
&\int\left\| \cT_{A\rightarrow B}(U_{A}\rho_{AE} U_{A}\mdag) - \widehat{\cT}_{A\rightarrow B}(U_{A}\rho_{AE} U_{A}\mdag) \right\|_1 dU\nonumber\\
&= \int \left\| (\mathcal{D}^{+}_{A\rightarrow B} - \mathcal{D}^{-}_{A\rightarrow B})(U_{A}\rho_{AE} U_{A}\mdag) \right\|_1 dU\nonumber\\
&\leqslant \int \left\| \mathcal{D}^{+}_{A\rightarrow B}(U_{A}\rho_{AE} U_{A}\mdag) \right\|_1 dU + \int \left\| \mathcal{D}^{-}_{A\rightarrow B}(U_{A}\rho_{AE} U_{A}\mdag) \right\|_1 dU\nonumber\\
&= \int \tr\left[\mathcal{D}^{+}_{A\rightarrow B}(U_{A}\rho_{AE} U_{A}\mdag)\right] dU + \int \tr\left[ \mathcal{D}^{-}_{A\rightarrow B}(U_{A}\rho_{AE} U_{A}\mdag) \right] dU\nonumber\\
&= \tr\left[\mathcal{D}^{+}_{A\rightarrow B}\left(\frac{\ident_{A}}{|A|}\otimes \rho_E\right)\right] + \tr\left[ \mathcal{D}^{-}_{A\rightarrow B}\left(\frac{\ident_{A}}{|A|}\otimes \rho_E\right) \right]\nonumber\\
&= \tr\left[\Delta^{+}_{A}\otimes \rho_E\right]+\tr\left[\Delta^{-}_{A}\otimes\rho_E\right]\nonumber\\
&\leqslant 4\eps\ .
\end{align}
This results in
\begin{align}
\int\left\| \cT_{A\rightarrow B}(U_{A}\rho_{AE} U_{A}\mdag) - \tau_{B}\otimes\rho_{E}\right\|_1 dU \leqslant 2^{- \demi H_{\min}^{\eps}(A|E)_{\rho}-\demi H_{\min}^{\eps}(A|B)_{\tau}} + 12\eps\ .
\end{align}
\qed

%%%%%%%%%%%%%%%%%%%%%%%%%%%%%%%%%%%%%%%%%%%%%%%%%%%%%%%%%%%%%%%%%%%%%

\section{Converse} \label{sec:converse}

The main purpose of this section is to state and prove a theorem (Theorem~\ref{thm:converse}) which implies that the achievability result of the previous section (Theorem~\ref{thm:achievability}) is essentially optimal for many natural choices of the mapping $\cT$.

%%%%%%%%%%%%%%%%%%%%%%%%%%%%%%%%%%%%%%%%%%%%%%%%%%%%%%%%%%%%%%%%%%%%%

\subsection{Statement of the Converse Theorem}

According to Theorem~\ref{thm:achievability}, decoupling is achieved whenever the term $H_{\min}^{\eps}(A|E)_{\rho}+H_{\min}^{\eps}(A|B)_{\tau}$ is sufficiently larger than $0$. Our converse now says that this is also a necessary condition (up to additive terms of the order $\log(1/\eps)$ and the scaling of the smoothing parameter) if one replaces the smooth conditional min-entropy in the second term, $H_{\min}^{\eps}(A|B)_{\tau}$ (which characterizes the channel), by a smooth conditional max-entropy.

\begin{thm}[Decoupling Converse]\label{thm:converse}
Let $\rho_{AE}\in\cS_{=}(\cH_{AE})$, $\mathcal{T}_{A \rightarrow B}$ be a TPCPM, and suppose that
\begin{align}
\left\|\mathcal{T}_{A\rightarrow B}(\rho_{AE})-\mathcal{T}_{A\rightarrow B}(\rho_A)\otimes\rho_E\right\|_1\leqslant\varepsilon\ .
\end{align}
Then, we have for any $\varepsilon',\varepsilon''>0$ that
\begin{align}
H_{\min}^{2\sqrt{6\varepsilon''+2\varepsilon}+2\sqrt{\varepsilon'}+\varepsilon''}(A|E)_{\rho}+H_{\max}^{ε''}(A|B)_{\omega}\geqslant-\log\frac{1}{\varepsilon'}\ ,
\end{align}
where $\omega_{AB}=\cT_{A'\rightarrow B}(\rho_{AA'})$ with $\rho_{AA'}\in\cS_{=}(\cH_{AA'})$ a purification of $\rho_{A}$, and $\cH_{A'}\cong\cH_{A}$.
\end{thm}

Note that we could also write $\omega_{AB}=|A|\left(\sqrt{\rho_A}\right)^\tps J(\mathcal{T}) \left(\sqrt{\rho_A}\right)^{\tps}$. In our formulation of the converse theorem, the mapping $\cT$ is not necessarily prepended by a unitary and the state that appears in the entropy term of the TPCPM is given by the more general expression $\omega_{AB}=\cT_{A'\rightarrow B}(\rho_{AA'})$ (rather than $\tau_{AB}=J(\cT)$ as in Theorem~\ref{thm:achievability}, corresponding to the case where $\rho_A$ is fully mixed). However, if we apply the converse to a TPCPM of the form $\bar{\cT}=\cT\circ\cU$, where $\cU$ corresponds to a random unitary channel applied to the input, Theorem~\ref{thm:converse} simplifies to the following.

\begin{cor}\label{cor:conv}
For the same premises as in Theorem~\ref{thm:converse}, but applied to the TPCPM $\bar{\cT}_{A\rightarrow B}=\cT_{A\rightarrow B}\circ\cU_{A}$, where $\cU_{A}$ corresponds to a Haar random unitary channel applied to the input, we have that
\begin{align}
H_{\min}^{2\sqrt{6\varepsilon''+2\varepsilon}+2\sqrt{\varepsilon'}+\varepsilon''}(A|E)_{\rho}+H_{\max}^{ε''}(A|B)_{\tau}\geqslant-\log\frac{1}{\varepsilon'}\ ,
\end{align}
where $\tau_{AB}=J(\cT)$.
\end{cor}

\begin{proof}
By assumption we have
\begin{align}
\int_{\mbU(A)}\left\|\cT_{A\rightarrow B}(U_{A}\rho_{AE}U_{A}^{\dagger})-\cT_{A\rightarrow B}(U_{A}\rho_A U_{A}^{\dagger})\otimes\rho_E\right\|_1 dU\leqslant\varepsilon\ .
\end{align}
and since the unitary $U_{A}$ is chosen at random, this is equivalent to
\begin{align}\label{eq:Fmap}
\left\|\cT_{A\rightarrow B}\circ\cF_{A\rightarrow AU}(\rho_{AE})-\cT_{A\rightarrow B}\circ\cF_{A\rightarrow AU}(\rho_A)\otimes\rho_E\right\|_1\leqslant\varepsilon\ ,
\end{align}
where $\cF_{A\rightarrow AU}$ denotes the TPCPM that chooses at random a unitary $U_{A}$ and outputs the choice of $U_{A}$. Now, let $\sigma_{AUER}$ be a purification of $\sigma_{AUE}=\cF_{A\rightarrow AU}(\rho_{AE})$ and note that $\sigma_{A}=\frac{\ident_{A}}{|A|}$ as well as $\sigma_{E}=\rho_{E}$. We apply Theorem~\ref{thm:converse} to~\eqref{eq:Fmap} with the map $\cT_{A\rightarrow B}$ and the state $\sigma_{AUE}$ to get
\begin{align}\label{eq:new_converse}
H_{\min}^{\delta}(A|UE)_{\sigma}+H_{\max}^{ε''}(UER|B)_{\cT(\sigma)}\geqslant-\log\frac{1}{\varepsilon'}\ ,
\end{align}
for $\delta=2\sqrt{6\varepsilon''+2\varepsilon}+2\sqrt{\varepsilon'}+\varepsilon''$. Since the state $\sigma_{AUER}$ and the maximally entangled state $\ket{\Phi}\bra{\Phi}_{AA'}$ are both purifications of $\frac{\ident_{A}}{|A|}$, there exists by Uhlmann's theorem~\cite{uhlmann} an isometry $\cW_{UER\rightarrow A'}$ such that $\ket{\Phi}\bra{\Phi}_{AA'}=\cW_{UER\rightarrow A'}(\sigma_{AUER})$. Hence, we have that $\cT_{A\rightarrow B}(\ket{\Phi}\bra{\Phi}_{AA'})=\cW_{UER\rightarrow A'}\circ\cT_{A\rightarrow B}(\sigma_{AUER})$, and by the invariance of the smooth conditional max-entropy under local isometries (Lemma \ref{lem:isometry}) we get
\begin{align}
H_{\max}^{ε''}(UER|B)_{\cT(\sigma)}=H_{\max}^{ε''}(A'|B)_{\cT(\ket{\Phi}\bra{\Phi})}=H_{\max}^{ε''}(A|B)_{\tau}\ .
\end{align}
Finally, we show that $H_{\min}^{\delta}(A|UE)_{\sigma}$ in~\eqref{eq:new_converse} is upper bounded by $H_{\min}^{\delta}(A|E)_{\rho}$. Since the register $U$ in $\sigma_{AUE}$ is classical, we can copy $U$ to another register $U'$ resulting in the state $\sigma_{AUU'E}$. With Lemma~\ref{lem:cc_equal} we then have
\begin{align}
H_{\min}^{\delta}(A|UE)_{\sigma}=H_{\min}^{\delta}(AU'|UE)_{\sigma}\ .
\end{align}
But now there exists an isometry $\cV_{AU'\rightarrow A}$ that reverses the action of the TPCPM $\cF$ such that $\cV_{AU'\rightarrow A}(\sigma_{AU'E})=\rho_{AE}$ (we let $\cV$ act on the copy $U'$ instead of $U$). Using the data processing inequality for the smooth conditional min-entropy (Lemma \ref{lem:strong}) and the invariance of the smooth conditional min-entropy under local isometries (Lemma \ref{lem:isometry}), we conclude
\begin{align}
H_{\min}^{\delta}(AU'|UE)_{\sigma}\leq H_{\min}^{\delta}(AU'|E)_{\sigma}= H_{\min}^{\delta}(A|E)_{\rho}\ .
\end{align}
\end{proof}

It can also be verified that the two terms, $H_{\min}^{\eps}(A|B)_{\tau}$ (from the achievability in Theorem~\ref{thm:achievability}) and $H_{\max}^{\eps}(A|B)_{\tau}$ (from the converse in Corollary~\ref{cor:conv}), coincide whenever the relevant states are essentially flat (i.e., proportional to projectors). This is the case for many channels used in applications (e.g., for state merging, cf.~Section~\ref{sec:applications}). Examples of such channels are given in Table~\ref{tb:mapping}. Furthermore, as we shall explain in the discussion section (Section~\ref{sec:discussion}), the two terms coincide asymptotically for iid channels.

%%%%%%%%%%%%%%%%%%%%%%%%%%%%%%%%%%%%%%%%%%%%%%%%%%%%%%%%%%%%%%%%%%%%%

\subsection{Proof of the Converse Theorem (Theorem~\ref{thm:converse})}

Let $\rho_{AER}$ be a purification of $\rho_{AE}$, $\cW_{A\rightarrow BB'}$ a Stinespring dilation~\cite{Stinespring55} of $\cT_{A\rightarrow B}$ and define
\begin{align}
\tilde{\sigma}_{BB'ER}=\ket{\tilde{\sigma}}\bra{\tilde{\sigma}}_{BB'ER}=\cW_{A\rightarrow BB'}(ρ_{AER})\ .
\end{align}
We have by Uhlmann's theorem~\cite{uhlmann} that $\omega_{AB}$ and $\tilde{\sigma}_{BER}$ are related by an isometry $\cV_{A\rightarrow ER}$, and hence by the invariance of the smooth conditional max-entropy under local isometries (Lemma~\ref{lem:isometry}) that
\begin{align}
H^{ε''}_{\max}(A|B)_{\omega}=H_{\max}^{ε''}(ER|B)_{\tilde{\sigma}}\ .
\end{align}
Furthermore, let $σ_{BB'ER}=\ket{σ}\bra{σ}_{BB'ER}$ be a subnormalized state with $P(σ,\tilde{σ}) ≤ ε''$ such that $H_{\max}(ER|B)_{\sigma}=H^{ε''}_{\max}(A|B)_{\omega}$, as well as $\sigma_{BB'ER}=\ket{\bar{\sigma}}\bra{\bar{\sigma}}_{BB'ER}$ such that $\bar{\sigma}_{BE} = \sigma_B \otimes \sigma_E$ and
\begin{align}
F(\sigma_{BB'ER}, \bar{\sigma}_{BB'ER}) = F(\sigma_{BE}, \sigma_B \otimes \sigma_E)\ .
\end{align}
Such a state exists by Uhlmann's theorem~\cite{uhlmann}, and can be shown to satisfy $P(\bar{\sigma},\sigma) \leqslant \sqrt{6\varepsilon''+2\varepsilon}$. The latter bound is obtained from 
\begin{align}
\| \sigma_{BE} - \bar{\sigma}_{BE} \|_1 &\leqslant \| \sigma_{BE} - \tilde{\sigma}_{BE} \|_1 + \| \tilde{\sigma}_{BE} - \bar{\sigma}_{BE} \|_1\nonumber\\
&\leqslant \| \sigma_{BE} - \tilde{\sigma}_{BE} \|_1 + \| \tilde{\sigma}_{BE} - \tilde{\sigma}_B \otimes \tilde{\sigma}_E \|_1 + \| \tilde{\sigma}_B \otimes \tilde{\sigma}_E - \sigma_B \otimes \sigma_E \|_1\nonumber\\
&\leqslant \varepsilon'' + \varepsilon + \| \tilde{\sigma}_B \otimes \tilde{\sigma}_E - \tilde{\sigma}_B \otimes \sigma_E \|_1 + \| \tilde{\sigma}_B \otimes \sigma_E - \sigma_B \otimes \sigma_E \|_1\nonumber\\
&\leqslant 3\varepsilon'' + \varepsilon\ ,
\end{align}
combined with the equivalence of purified distance and trace distance (Lemma~\ref{lem:purified-dist-vs-trace-dist}). Now, we know from a technical lemma about the conditional max-entropy (Lemma~\ref{lem:hmax-optimal-z}) that
\begin{align}
\sigma_{BB'ER} \leqslant 2^{\hmax(ER|B)_{\sigma|\sigma}}\cdot Y_{EBR} \otimes \ident_{B'}\ ,
\end{align}
where
\begin{align}
Y_{BER} = 2^{-\demi \hmax(ER|B)_{\sigma|\sigma}}\cdot\sigma_B^{-1/2} \sqrt{\sigma_B^{1/2} \sigma_{BER} \sigma_B^{1/2}} \sigma_{B}^{-1/2}\ .
\end{align}
This implies that
\begin{align}\label{eqn:opeq-sigma}
\sigma_{BB'ER} \leqslant \frac{2^{\hmax(ER|B)_{\sigma|\sigma}}}{\varepsilon'}\cdot\left((1-\varepsilon')\cdot\sigma_B^{-1/2}\bar{\sigma}_{BER}\sigma_B^{-1/2}+\varepsilon'\cdot Y_{BER} \right) \otimes \ident_{B'}
\end{align}
for any $\varepsilon'>0$. Tracing out the $R$ system, we get
\begin{align}
\sigma_{BEB'} \leqslant \frac{2^{\hmax(ER|B)_{\sigma|\sigma}}}{\varepsilon'}\cdot\left((1-\varepsilon')\cdot\ident_B\otimes\sigma_E+\varepsilon'\cdot Y_{BE}\right)\otimes\ident_{B'}\ .
\end{align}
We now define $G_{BE}=\sqrt{1-\varepsilon'}\cdot\sigma_E^{1/2}((1-\varepsilon')\cdot\ident_{B}\otimes\sigma_E+\varepsilon'\cdot Y_{BE})^{-1/2}$. Note that $G$ is a contraction, i.e., $\| G \|_{\infty}\leqslant 1$,
\begin{align}
G G^{\dagger} &= (1-\varepsilon')\cdot\sigma_E^{1/2} \left( (1-\varepsilon')\cdot\ident_B \otimes \sigma_E + \varepsilon'\cdot Y_{BE} \right)^{-1} \sigma_E^{1/2}\nonumber\\
&\leqslant (1-\varepsilon')\cdot\sigma_E^{1/2} \left( (1-\varepsilon')\cdot\ident_B \otimes \sigma_E \right)^{-1} \sigma_E^{1/2}\nonumber\\
&= \ident_{BE}\ ,
\end{align}
where we have used the operator monotonicity of $f(t) = -1/t$.  At this point, we conjugate both sides of~\eqref{eqn:opeq-sigma} by $G_{BE}$ to get
\begin{align}
G_{BE} \sigma_{BEB'} G_{BE}^{\dagger} &\leqslant \frac{(1-\varepsilon')\cdot2^{\hmax(ER|B)_{\sigma|\sigma}}}{\varepsilon'}\cdot\sigma_E \otimes \ident_{BB'}\nonumber\\
&\leqslant\frac{2^{\hmax(ER|B)_{\sigma|\sigma}}}{\varepsilon'}\cdot\sigma_E \otimes \ident_{BB'}\ .\label{eqn:right-after-G}
\end{align}
Let us now define $\ket{\psi}_{BERB'}=G_{BE} \ket{\sigma}_{BERB'}$ and note that $\psi_{BERB'}=\ket{\psi}\bra{\psi}_{BERB'}$ is a subnormalized state since $G$ is a contraction. Then, we can rewrite
(\ref{eqn:right-after-G}) as
\begin{align}
\psi_{BEB'}\leqslant\frac{2^{\hmax(ER|B)_{\sigma|\sigma}}}{\varepsilon'}\cdot\sigma_E \otimes \ident_{BB'}\ ,
\end{align}
which implies
\begin{align}
\hmin(BB'|E)_{\psi|\sigma} \geqslant -\hmax(ER|B)_{\sigma|\sigma} - \log(1/\varepsilon')\ .
\end{align}
We will now need to show that $\psi_{BEB'}$ is $(2\sqrt{6\varepsilon''+2\varepsilon}+2\sqrt{\varepsilon'}+\varepsilon'')$-close to $\tilde{\sigma}_{BEB'}$, because the invariance of the smooth conditional min-entropy under local isometries (Lemma~\ref{lem:isometry}) then implies the claim
\begin{align}
\hmin^{2\sqrt{6\varepsilon''+2\varepsilon}+2\sqrt{\varepsilon'}+\varepsilon''}(A|E)_{\rho} + \hmax^{\varepsilon''}(A|B)_{\omega} \geqslant  - \log(1/\varepsilon')\ .
\end{align}
To this end, we shall define the following vectors
\begin{align}
&\ket{\psi'}_{BERB'}=G_{BE}^{\dagger}\ket{\bar{\sigma}}_{BERB'}\\
&\ket{\psi''}_{BERB'}=G_{BE}\ket{\bar{\sigma}}_{BERB'}\\
&\ket{\tilde{\psi}}_{BERB'}=\sqrt{1-\varepsilon'}\cdot G_{BE}^{-1} \ket{\bar{\sigma}}_{BERB'}\ .
\end{align}
We first show that all these vectors define subnormalized states such that the purified distance between them is well-defined. Since $G_{BE}$ is a contraction, we immediately get that $\| \ket{\psi'}_{BERB'}\| \leqslant 1$ and $\| \ket{\psi''}_{BERB'}\| \leqslant 1$. Furthermore, we have that
\begin{align}
\left\| \ket{\tilde{\psi}}_{BERB'} \right\|^2 &= (1-\varepsilon')\cdot\bra{\bar{\sigma}} {G_{BE}^{-1}}^{\dagger} G_{BE}^{-1} \ket{\bar{\sigma}}\nonumber\\
&= \bra{\bar{\sigma}} \sigma_E^{-1/2} \left( (1-\varepsilon')\cdot\ident_{B}\otimes\sigma_E + \varepsilon' Y_{BE} \right) \sigma_{E}^{-1/2} \ket{\bar{\sigma}}\nonumber\\
&= 1-\varepsilon' + \varepsilon'\cdot\bra{\bar{\sigma}} \sigma_E^{-1/2} Y_{BE} \sigma_E^{-1/2} \ket{\bar{\sigma}}\nonumber\\
&= 1-\varepsilon' + \varepsilon'\cdot\tr\left[ Y_{BE} \sigma_E^{-1/2} \bar{\sigma}_{EB} \sigma_E^{-1/2} \right]\nonumber\\
&= 1-\varepsilon' + \varepsilon'\cdot\tr\left[ Y_{BE} \sigma_B\right]\nonumber\\
&= 1\ .
\end{align}
We have $\braket{\tilde{\psi}}{\psi'} = \sqrt{1-\varepsilon'}$, and
\begin{align}
\braket{\bar{\sigma}}{\tilde{\psi}} &= \sqrt{1-\varepsilon'}\cdot\bra{\bar{\sigma}} G_{BE}^{-1}  \ket{\bar{\sigma}}\nonumber\\
&= \tr\left[ (\sigma_B \otimes \sigma_E) \left( (1-\varepsilon')\cdot\ident_{B}\otimes\sigma_E + \varepsilon' Y_{BE} \right)^{1/2} \sigma_E^{-1/2} \right]\nonumber\\
&= \tr\left[ (\sigma_B \otimes \sigma_E^{1/2}) \left( (1-\varepsilon')\cdot\ident_{B}\otimes\sigma_E + \varepsilon' Y_{BE} \right)^{1/2}  \right]\nonumber\\
&\geqslant \tr\left[ (\sigma_B \otimes \sigma_E^{1/2})\cdot\sqrt{1-\varepsilon'}\cdot(\ident_{B}\otimes\sigma_E^{1/2})\right]\nonumber\\
&= \sqrt{1-\varepsilon'}\cdot\tr\left[ \sigma_B \otimes \sigma_E  \right]\nonumber\\
&= \sqrt{1-\varepsilon'}\ ,
\end{align}
where the inequality is due to the operator monotonicity of the square-root function. Therefore, we have that $P(\psi', \bar{\sigma}) \leqslant 2\sqrt{\varepsilon'}$ and furthermore $P(\psi'', \bar{\sigma}) = P(\psi', \bar{\sigma})$, since
\begin{align}
F(\psi'', \bar{\sigma}) = \bra{\bar{\sigma}} G_{BE}^{\dagger} \ket{\bar{\sigma}} = F(\bar{\sigma}, \psi')\ .
\end{align}
Since conjugation by $G$ is trace-non-increasing, we also have $P(\psi'', \psi) \leqslant P(\sigma, \bar{\sigma})\leq\sqrt{6\varepsilon''+2\varepsilon}$. This implies
\begin{align}
P(\psi, \tilde{\sigma})&\leqslant P(\psi, \psi'') + P(\psi'', \bar{\sigma})+P(\bar{\sigma},\sigma)+P(\sigma,\tilde{\sigma})\nonumber\\
&\leqslant \sqrt{6 \varepsilon'' + 2\varepsilon} +2 \sqrt{\varepsilon'}+\sqrt{6 \varepsilon'' + 2\varepsilon}+\varepsilon''\ .
\end{align}
\qed

%%%%%%%%%%%%%%%%%%%%%%%%%%%%%%%%%%%%%%%%%%%%%%%%%%%%%%%%%%%%%%%%%%%%%

\section{One-Shot State Merging}\label{sec:applications}

As an example application of the decoupling theorem and its converse
we discuss one-shot quantum state merging. This is a two-party
task: its goal is to transfer the information contained in a quantum
system, $A$, initially held by one party, Alice, to the other party,
Bob. This should be achieved with only limited resources (such as
entanglement or communication). It is taken into account that Bob may
have access to a quantum system, $B$, correlated to $A$, which may be
used to minimize the use of resources. The term one-shot is
used to emphasize that the task is considered in the general one-shot
scenario. As explained in the discussion section, the asymptotic iid
results, where many independent copies of a given state are
transferred, can be recovered as a special case.

The notion of quantum state merging has been introduced
in~\cite{HoOpWi05Nat,HoOpWi07CMP} and a protocol has been proposed
that achieves the task in the asymptotic iid scenario.  The more
general one-shot setup we consider here was first analyzed
in~\cite{diploma-berta} and preliminary results appeared
in~\cite{min-max-entropy}.

We start giving a formal definition of quantum state
merging~\cite{HoOpWi05Nat,HoOpWi07CMP,diploma-berta}. Let $\rho_{A B}$
be the joint initial state of Alice and Bob's systems. We can view
this state as part of a larger pure state $\rho_{ABE}$ that includes a
reference system $E$. In this picture state merging means that Alice
can send the $A$-part of $\rho_{ABE}$ to Bob's side without altering
the joint state. We consider the particular setting proposed
in~\cite{HoOpWi05Nat} where classical communication from Alice to Bob
is free, but no quantum communication is possible.  Furthermore, Alice
and Bob have access to a source of entanglement and their goal is to
minimize the number of entangled bits consumed during the protocol (or
maximize the number of entangled bits that can be generated).

\begin{defin}[Quantum State Merging]
Let $\rho_{AB}\in\cS_{=}(\cH_{AB})$, and let $A_{0}B_{0}$ be additional systems. A TPCPM $\mathcal{E}:AA_{0}\otimes BB_{0}\rightarrow A_{1}\otimes B_{1}B'B$ is called quantum state merging of $\rho_{AB}$ with error $\eps\geq0$, if it is a local operation and classical forward communication process for the bipartition $AA_{0}\rightarrow A_{1}$ vs.~$BB_{0}\rightarrow B_{1}B'B$, and
\begin{align}
(\mathcal{E}_{AA_{0}BB_{0}\rightarrow A_{1}B_{1}B'B})(\Phi^{K}_{A_{0}B_{0}}\otimes\rho_{ABE})\approx_{\eps}\Phi^{L}_{A_{1}B_{1}}\otimes\rho_{BB'E}\ ,
\end{align}
where $\rho_{BB'E}=(\opid_{A\rightarrow
  B'}\otimes\opid_{BE})\rho_{ABE}$ for a purification $\rho_{ABE}$ of
$\rho_{AB}$, and $\Phi^{K}$, $\Phi^{L}$ are maximally entangled states
on $A_{0}B_{0}$, $A_{1}B_{1}$ of Schmidt-rank $K$ and $L$,
respectively. The number
$$l^{\eps}=\log{K}-\log{L}$$
is called entanglement cost.\footnote{In the original references~\cite{HoOpWi05Nat,HoOpWi07CMP} quantum state merging was defined slightly differently, namely as a local operation and classical two-way communication process. However, their protocol for the achievability only uses classical forward communication.}
\label{defstate}
\end{defin}

We are interested in quantifying the minimal entanglement cost for
quantum state merging of $\rho_{AB}$ with error $\eps$. For this, we
use the achievability and converse for decoupling
(Theorem~\ref{thm:achievability} and
Theorem~\ref{thm:converse}). These allow us to derive essentially
tight (up to additive terms of the order $\log(1/\eps)$ and the scaling of the smoothing parameter) bounds on the entanglement cost.

The basic idea underlying our analysis of quantum state merging is the
observation that the desired situation after the protocol execution is
necessarily such that Alice's system is decoupled from the
reference. Furthermore, it follows from Uhlmann's theorem~\cite{uhlmann} that this decoupling is also sufficient.

\begin{thm}[Achievability for Quantum State Merging]\label{thm:achiev}
The minimal entanglement cost for quantum state merging of $\rho_{AB}\in\cS_{=}(\cH_{AB})$ with error $\eps>0$ is upper bounded by
\begin{align}
l^{\eps}\leq H_{\max}^{\eps^{2}/13}(A|B)_{\rho}+4\log(1/\eps)+2\log13\ .
\end{align}
\end{thm}

\begin{proof}
Let $\rho_{A B E}$ be a purification of $\rho_{A B}$. The intuition is as follows. In the first step of the protocol, Alice decouples her part from the reference (employing Theorem~\ref{thm:achievability}), where she chooses a rank-$L$ projective measurement as the TPCPM, and she sends the measurement result to Bob. For all measurement outcomes the post-measurement state on Alice's side is then approximately given by $\frac{\ident_{A_{1}}}{|A_{1}|}\otimes\rho_{E}$ and Bob holds a purification of this. But $\frac{\ident_{A_{1}}}{|A_{1}|}\otimes\rho_{E}$ is the reduced state of $\Phi^{L}_{A_{1}B_{1}}\otimes\rho_{BB'E}$ as well and since all purifications are equal up to local isometries, there exists an isometry on Bob's side that transform the state into $\Phi^{L}_{A_{1}B_{1}}\otimes\rho_{BB'E}$ (by Uhlmann's theorem~\cite{uhlmann}); this is then the second step of the protocol.

More formally, choose $K$ and $L$ such that
\begin{align}
\log K-\log L=H_{\max}^{\eps^{2}/13}(A|B)_{\rho} + 4\log(1/\eps) + 2\log13\ ,
\label{KL}
\end{align}
which is the entanglement cost of the protocol.\footnote{Since we need $K$, $L\in\mathbb{N}$, we can not choose $\log K-\log L$ exactly equal to $H_{\max}^{\eps^{2}/13}(A|B)_{\rho}+4\log(1/\eps)+2\log13$ in general. Rather, we need to choose $K$, $L\in\mathbb{N}$ such $\log K-\log L$ is minimal but still greater or equal than $H_{\max}^{\eps^{2}/13}(A|B)_{\rho}+4\log(1/\eps)+2\log13$.}

Choose $N$ fixed orthogonal subspaces of dimension $L$ on $AA_{0}$,\footnote{For simplicity assume that $K\cdot|A|$ is divisible by $L$. In general one has to choose $N-1$ fixed orthogonal subspaces of dimension $L$ and one of dimension $L'=K\cdot|A|-(N-1)\cdot L<L$. The proof remains the same, although some coefficients change.} denote the projectors on these subspaces followed by a fixed unitary mapping it to $A_{1}$ by $P^{x}_{A_{0}A\rightarrow A_{1}}$ and define the isometry
\begin{align}
W_{A_{0}A\rightarrow A_{1}X_{A}X_{B}}=\sum_{x}P^{x}_{A_{0}A\rightarrow A_{1}}\otimes\ket{x}_{X_{A}}\otimes\ket{x}_{X_{B}}\ .
\label{proj}
\end{align}
Denote by $U_{A_{0}A}$ a unitary selected randomly according to the Haar measure over the unitary group on $\cH_{A_{0}A}$ and write
\begin{align}
\theta_{A_0 B_0 ABE} &= \Phi^K_{A_0B_0} \otimes \rho_{ABE}\\
\sigma_{A_{0}B_{0}ABE} &= U_{A_{0}A} \theta_{A_0 B_0 ABE} U_{A_{0}A}^{\dagger}.
\end{align}
Now the first step of the protocol is to apply this unitary followed by the isometry~\eqref{proj}, and to send the $X_B$ system to Bob. In order to take into account that the channel is classical, we keep a copy $X_{A}$ at Alice's side.

By the decoupling theorem (Theorem~\ref{thm:achievability}) we get for
\begin{align}
\sigma_{A_{1}X_{A}X_{B}B_{0}BE}=(W_{A_{0}A\rightarrow A_{1}X_{A}X_{B}})\sigma_{A_{0}B_{0}ABE}(W_{A_{0}A\rightarrow A_{1}X_{A}X_{B}})^{\dagger}\ .
\end{align}
that
\begin{align}
\|\sigma_{A_{1}X_{A}E}-\tau_{A_{1}X_{A}}\otimes\rho_{E}\|_{1}\leq2^{-1/2(H_{\min}^{\eps^{2}/13}(A_{0}A|E)_{\theta}+H_{\min}^{\eps^{2}/13}(A_{0}'A'|A_{1}X_{A})_{\tau})}+\frac{12\eps^{2}}{13}\ ,
\label{dec}
\end{align}
where $A_{0}'A'$ is a copy of $A_{0}A$, and
\begin{align}
\ket{\tau}_{A_{0}'A'A_{1}X_{A}X_{B}}=W_{A_{0}A\rightarrow A_{1}X_{A}X_{B}}\ket{\Phi}_{A_{0}'A'A_{0}A}
\end{align}
with
\begin{align}
\ket{\Phi}_{A_{0}'A'A_{0}A}=\frac{1}{K\cdot|A|}\sum_{i}\ket{i}_{A_{0}'A'}\otimes\ket{i}_{A_{0}A}\ .
\end{align}
We can simplify this using the superadditivity of the smooth conditional min-entropy (Lemma~\ref{super}) and the duality between smooth conditional min- and max-entropy (Lemma \ref{lem:dual})
\begin{align}
H_{\min}^{\eps^{2}/13}(A_{0}A|E)_{\theta}\geq H_{\min}^{\eps^{2}/13}(A|E)_{\rho}+\log K=-H_{\max}^{\eps^{2}/13}(A|B)_{\rho}+\log K\ .
\label{step1}
\end{align}
Furthermore, because $\tau_{A_{0}'A'A_{1}X_{A}}$ is classical on $X_{A}$, we can use a lemma about the conditional min-entropy of classical-quantum states (Lemma~\ref{classical}) and get
\begin{align}
H_{\min}^{\eps^{2}/13}(A_{0}'A'|A_{1}X_{A})_{\tau}&\geq H_{\min}(A_{0}'A'|A_{1}X_{A})_{\tau}\nonumber\\
&=-\log(\sum_{x}p_{x}\cdot2^{-H_{\min}(A_{0}'A'|A_{1})_{\tau^{x}}})\nonumber\\
&\geq\min_{x}H_{\min}(A_{0}'A'|A_{1})_{\tau^{x}}\ ,
\end{align}
where
\begin{align}
&\tau^{x}_{A_{0}'A'A_{1}}=\frac{1}{\sqrt{p_{x}}}P_{A_{0}A\rightarrow A_{1}}^{x}\ket{\Phi}_{A_{0}'A'A_{0}A}\\
&p_{x}=\|P_{A_{0}A\rightarrow A_{1}}^{x}\ket{\Phi}_{A_{0}'A'A_{0}A}\|\ .
\end{align}
But since $P_{A_{0}A\rightarrow A_{1}}^{x}$ is a rank $L$ projector, we can use a dimension lower bound of the conditional min-entropy (Lemma~\ref{lem:dim}) to conclude that for all $x$
\begin{align}
H_{\min}(A_{0}'A'|A_{1})_{\tau^{x}}\geq-\log L\ .
\end{align}
This together with~\eqref{KL},~\eqref{dec} and~\eqref{step1} implies
\begin{align}
\left\|\sigma_{A_{1}X_{A}E}-\frac{\ident_{A_{1}}}{|A_{1}|}\otimes\tau_{X_{A}}\otimes\rho_{E}\right\|_{1}&=\left\|\sigma_{A_{1}X_{A}E}-\tau_{A_{1}X_{A}}\otimes\rho_{E}\right\|_{1}\nonumber\\
&\leq2^{-1/2(\log K-\log L-H_{\max}^{\eps^{2}/13}(A|B)_{\rho})}+\frac{12\eps^{2}}{13}\nonumber\\
&=2^{-1/2(4\log(1/\eps)+2\log13)}+\frac{12\eps^{2}}{13}=\eps^{2}\ ,
\end{align}
and hence $F(\sigma_{A_{1}X_{A}E},\frac{\ident_{A_{1}}}{|A_{1}|}\otimes\tau_{X_{A}}\otimes\rho_{E})\geq1-\eps^{2}/2$ (by Lemma~\ref{lem:purified-dist-vs-trace-dist}).

In the second step of the protocol, Bob decodes the system to the state $\rho_{BB'E} \otimes \Phi_{A_1B_1}$. A suitable decoder can be shown to exist using Uhlmann's theorem~\cite{uhlmann}. There exists an isometry $\cV_{BB_{0}X_{B}\rightarrow BB'B_{1}X_{B}}$ such that for
\begin{align}
\eta_{A_{1}X_{A}X_{B}BB'B_{1}E}=\cV_{BB_{0}X_{B}\rightarrow BB'B_{1}X_{B}}(\sigma_{A_{1}X_{A}X_{B}BB_{0}E})
\end{align}
\begin{align}
F(\sigma_{A_{1}X_{A}E},\frac{\ident_{A_{1}}}{|A_{1}|}\otimes\tau_{X_{A}}\otimes\rho_{E})=F(\eta_{A_{1}X_{A}X_{B}BB'B_{1}E},\tau_{X_{A}X_{B}}\otimes\Phi^{L}_{A_{1}B_{1}}\otimes\rho_{BB'E})\ ,
\end{align}
and with that
\begin{align}\label{lem:fidend}
F(\eta_{A_{1}X_{A}X_{B}BB'B_{1}E},\tau_{X_{A}X_{B}}\otimes\Phi^{L}_{A_{1}B_{1}}\otimes\rho_{BB'E})\geq1-\frac{\eps^{2}}{2}\ .
\end{align}
Expressing this in the purified distance (with Lemma~\ref{lem:purified-dist-vs-trace-dist}) and discarding $X_{A}X_{B}$, we obtain a $\eps$-error quantum state merging protocol for $\rho_{ABE}$.
\end{proof}

\begin{thm}[Converse for Quantum State Merging]\label{thm:converse-state-merging}
The minimal entanglement cost for quantum state merging of $\rho_{AB}\in\cS_{=}(\cH_{AB})$ with error $\eps>0$ is lower bounded by
\begin{align}
l^{\eps}\geq H_{\max}^{4\sqrt{2\eps}+3\eps}(A|B)_{\rho}-2\log\frac{1}{\eps}\ .
\end{align}
\end{thm}

\begin{proof}
We start with noting that any $\eps$-error quantum state merging
protocol for $\rho_{AB}$ can be assumed to have the following form:
applying local operations at Alice's side, then sending a classical
register from Alice to Bob, and finally applying local operations at
Bob's side. For a purified state $\rho_{A B E}$,  the protocol
produces a state $\eps$-close to ${\Phi^{L}_{A_{1}B_{1}}\otimes\rho_{BB'E}}$.

As can be seen from the definition, it is a necessary step for any
quantum state merging protocol to decouple Alice's part from the
reference. The idea of the proof is to use the converse for decoupling
(Theorem~\ref{thm:converse}). This then results in the desired
converse for quantum state merging.

More precisely, a general $\eps$-error quantum state merging protocol for $\rho_{ABE}$ has the following form. At first some TPCPM
\begin{align}
\cT_{A_{0}A\rightarrow A_{1}X_{B}}(.)=\sum_{x}M^{x}_{A_{0}A\rightarrow A_{1}}(.)\otimes\ket{x}\bra{x}_{X_{B}}
\end{align}
is applied to the input state $\Phi^{K}_{A_{0}B_{0}}\otimes\rho_{ABE}$. By the Stinespring dilation~\cite{Stinespring55} we can think of this TPCPM as an isometry
\begin{align}\label{eq:newcptp}
W_{A_{0}A\rightarrow A_{1}A_{G}X_{B}X_{A}}=\sum_{x}M^{x}_{A_{0}A\rightarrow A_{1}A_{G}}\otimes\ket{x}_{X_{A}}\otimes\ket{x}_{X_{B}}\ ,
\end{align}
where the $M^{x}_{A_{0}A\rightarrow A_{1}A_{G}}$ are partial isometries and $A_{G},X_{A}$ are additional \lq garbage\rq~registers on Alice's side that will be discarded in the end. The isometry $W$ results in the state
\begin{align}
\ket{\gamma}_{A_{1}A_{G}X_{A}X_{B}BB_{0}E}=\sum_{x}\ket{\gamma^{x}}_{A_{1}A_{G}BB_{0}E}\otimes\ket{x}_{X_{A}}\otimes\ket{x}_{X_{B}}\ ,
\end{align}
with
\begin{align}
\ket{\gamma^{x}}_{A_{1}A_{G}BB_{0}E}=M^{x}_{A_{0}A\rightarrow A_{1}A_{G}}(\ket{\Phi^{K}}_{A_{0}B_{0}}\otimes\ket{\rho}_{ABE})\ .
\end{align}
The next step of the protocol is then to send the classical register $X_{B}$ to Bob.

Now let us analyze how the state $\gamma_{A_{1}A_{G}X_{A}E}$ has to
look like. By the definition of quantum state merging
(Definition~\ref{defstate}) the state at the end of the protocol has
to be $\eps$-close to $\Phi^{L}_{A_{1}B_{1}}\otimes\rho_{BB'E}$. This
implies that Alice's part $A_{1}$ has to be decoupled from the
reference. But because the state
$\Phi^{L}_{A_{1}B_{1}}\otimes\rho_{BB'E}$ is pure this also implies
that all additional registers, that we might have at the end of the
protocol, have to be decoupled as well. Thus we need
\begin{align}\label{eq:additional}
\gamma_{A_{1}A_{G}X_{A}E}\approx_{\eps}\frac{\ident_{A_{1}}}{|A_{1}|}\otimes\gamma_{A_{G}X_{A}}\otimes\rho_{E}\ ,
\end{align}
and in trace distance (using Lemma~\ref{lem:purified-dist-vs-trace-dist}) this reads
\begin{align}\label{eq:newmain}
\left\|\gamma_{A_{1}A_{G}X_{A}E}-\frac{\ident_{A_{1}}}{|A_{1}|}\otimes\gamma_{A_{G}X_{A}}\otimes\rho_{E}\right\|_{1}\leq2\eps\ .
\end{align}

Using the converse for decoupling (Theorem~\ref{thm:converse}) for the isometry $W_{A_{0}A\rightarrow A_{1}A_{G}X_{B}X_{A}}$ in \eqref{eq:newcptp} followed by the partial trace over $X_{B}$, we get that the decoupling condition~\eqref{eq:newmain} implies for any $\eps',\eps''>0$ that
\begin{align}
H_{\min}^{2\sqrt{6\eps''+2\eps}+2\sqrt{\eps'}+\eps''}(A_{0}A|E)_{\rho}+H_{\max}^{\eps''}(A_{0}'A'|A_{1}A_{G}X_{A})_{\omega}\geq&-\log\frac{1}{\eps'}\ ,
\end{align}
where
\begin{align}
\omega_{A_{0}'A'A_{1}A_{G}X_{A}}=\mathrm{tr}_{X_{B}}\left[(W_{A_{0}A\rightarrow A_{1}A_{G}X_{B}X_{A}})\zeta_{A_{0}'A'A_{0}A}(W^{\dagger}_{A_{0}A\rightarrow A_{1}A_{G}X_{B}X_{A}})\right]
\end{align}
for $\zeta_{A_{0}'A'A_{0}A}$ a purification of $\frac{\ident_{A_{0}}}{|A_{0}|}\otimes\rho_{A}$ with $A_{0}'A'$ a copy of $A_{0}A$. As a next step we simplify this in order to bring the converse into the desired form.

Choosing $\eps'=\eps^{2}$ and $\eps''=\eps$, using a dimension upper bound for the smooth conditional min-entropy (Lemma~\ref{old}), and the duality between smooth conditional min- and max-entropy (Lemma~\ref{lem:dual}) we obtain
\begin{align}
\log K+H_{\max}^{\eps}(A_{0}'A'|A_{1}A_{G}X_{A})_{\omega}\geq&H_{\max}^{4\sqrt{2\eps}+3\eps}(A|B)_{\rho}-2\log\frac{1}{\eps}\ .
\end{align}
By the decoupling criterion in purified distance (Equation~\eqref{eq:additional}), the state $\omega_{A_{0}'A'A_{1}A_{G}X_{A}}$ has to be $\eps$-close to a state
\begin{align}
\xi_{A_{0}'A'A_{1}A_{G}X_{A}}=\sum_{x}q_{x}\xi_{A_{0}'A'A_{1}A_{G}}^{x}\otimes\ket{x}\bra{x}_{X_{A}}\ ,
\end{align}
where $q_{x}$ is some probability distribution and $\xi_{A_{0}'A'A_{1}A_{G}}^{x}$ pure with $\xi_{A_{1}A_{G}}^{x}=\frac{\ident_{A_{1}}}{|A_{1}|}\otimes\xi_{A_{G}}^{x}$ for all $x$.
Hence
\begin{align}
H_{\max}^{\eps}(A_{0}'A'|A_{1}A_{G}X_{A})_{\omega}\leq H_{\max}(A_{0}'A'|A_{1}A_{G}X_{A})_{\xi}
\end{align}
and by a lemma about the conditional max-entropy of classical-quantum states (Lemma \ref{classical2})
\begin{align}
H_{\max}(A_{0}'A'|A_{1}A_{G}X_{A})_{\xi}=\log\left(\sum_{x}q_{x}\cdot2^{H_{\max}(A_{0}'A'|A_{1}A_{G})_{\xi^{x}}}\right)\ .
\end{align}
Using the duality between conditional min- and max-entropy (Lemma~\ref{lem:dual}) and a polar decomposition of $\xi^{x}_{A_{0}'A'A_{1}A_{G}}$, we get
\begin{align}
H_{\max}(A_{0}'A'|A_{1}A_{G})_{\xi^{x}}&=-H_{\min}(A_{0}'A')_{\xi^{x}}\nonumber\\
&=-H_{\min}(A_{1}A_{G})_{\xi^{x}}\nonumber\\
&=-H_{\min}(A_{1})_{\frac{\ident}{|A_{1}|}}-H_{\min}(A_{G})_{\xi^{x}}\nonumber\\
&\leq-H_{\min}(A_{1})_{\frac{\ident}{|A_{1}|}}\nonumber\\
&=-\log L\ .
\end{align}
Hence, the converse becomes
\begin{align}
\log K-\log L\geq H_{\max}^{4\sqrt{2\eps}+3\eps}(A|B)_{\rho}-2\log\frac{1}{\eps}\ .
\end{align}
\end{proof}

%Note that we essentially quantified the minimal entanglement cost for $\eps$-error quantum state merging of $\rho_{ABE}$ (up to additive terms of the order $\log(1/\eps)$).

%If we choose the input state to be independent and identically distributed, that is $\rho_{ABE}=\beta^{\otimes n}_{ABE}$, then it follows from the quantum asymptotic equipartition property (Lemma~\ref{thm:fully-quantum-aep}) that the entanglement cost rate in the asymptotic limit $n\rightarrow\infty$ for a vanishing error $\eps\rightarrow0$ converges to $H(A|B)_{\beta}$. This is also the scenario considered in the original papers~\cite{HoOpWi05Nat,HoOpWi07CMP}, where the same asymptotic bound has been proved.

%%%%%%%%%%%%%%%%%%%%%%%%%%%%%%%%%%%%%%%%%%%%%%%%%%%%%%%%%%%%%%%%%%%%%

\section{Discussion}\label{sec:discussion}

The main contribution of this work is a decoupling theorem, i.e., a sufficient (Theorem~\ref{thm:achievability}) and necessary (Theorem~\ref{thm:converse}) criterion for decoupling in terms of smooth conditional entropies. These criteria can then be applied to obtain tight characterizations of various operational tasks. As outlined in Section~\ref{sec:applications} by means of state merging, such applications are often possible because of a duality between independence and maximum entanglement: given a pure state $\rho_{BER}$ such that $\rho_B$ is maximally mixed, the property that the subsystem $B$ is independent of $E$ and the property that $B$ is fully entangled with $R$ are equivalent.

A crucial property of our decoupling criterion is that it gives (nearly optimal) bounds in a one-shot scenario, where the decoupling map $\cT$ may only be applied once (or, by replacing $\cT$ by $\cT^{\otimes k}$, any finite number of times). For a typical example, consider $m$ qubits, $A$, and assume that $A$ undergoes a reversible evolution, $\mathcal{U}$, after which we discard $m-m'$ qubits, corresponding to a partial trace, $\mathcal{T}= \tr_{m-m'}$ (see last example of Table~\ref{tb:mapping}). Our decoupling theorem (Theorem~\ref{thm:achievability}) then shows that decoupling up to an error $\eps$ is achieved for most choices of $\cU$ if
\begin{align}\label{eq:ex-tight}
m'\lessapprox\frac{1}{2}\left(m+H_{\min}^{\eps}(A|E)_{\rho}\right)\ .
\end{align}
In contrast to this, the original decoupling results~\cite{FQSW}, formulated in terms of smooth non-conditional entropies, only show that decoupling up to an error $\eps$ is achieved for most choices of $\cU$ if
\begin{align}\label{eq:ex-faroff}
m'\lessapprox\frac{1}{2}\left(m+H_{\min}^{\eps}(AE)_{\rho}-H_{\max}^{\eps}(E)_{\rho}\right)\ .
\end{align}
To see that this latter bound may be arbitrarily weaker than the bound~\eqref{eq:ex-tight} that uses smooth conditional entropies, consider the following completely classical state. Let $A$ and $E$ be perfectly correlated, and let the marginal distribution of $A$ (and $E$) have one value that is taken with probability 1/2, and be uniform over the remaining $2^{m}-1$ values. Then we have (for $\eps\geq0$ close to zero)
\begin{align}
H_{\min}^{\eps}(A|E)_{\rho}\approx0\quad\mathrm{vs.}\quad H_{\min}^{\eps}(AE)_{\rho}-H_{\max}^{\eps}(E)_{\rho}\approx1-m\ .
\end{align}
The difference between these two bounds is conceptually relevant. An example illustrating this is the quantitative Landauer's principle derived recently in~\cite{Faist12}. The result, which is based on the bound~\eqref{eq:ex-tight}, shows that correlations between the inputs and outputs of an irreversible mapping are relevant for the thermodynamic work cost of implementations of the mapping. These correlations  would not be accounted for if a bound of the form~\eqref{eq:ex-faroff} was used for the derivation of the principle.

In contrast to the original results on decoupling that are based on specific decoupling processes (where the mapping $\cT$ is either a partial trace~\cite{FQSW} or a projective measurement~\cite{HoOpWi07CMP}), our decoupling criterion is also applicable to general mappings $\cT$. This extension is, e.g., employed in~\cite[Section 5]{diploma-hutter} in order to discuss the postulate of equal a priori probability in quantum statistical mechanics.

Our generalizations of the decoupling technique are crucial for other applications in physics as well, e.g., for the analysis of thermodynamic systems~\cite{RARDV10}, for finding an efficient classical description of 1D quantum states with an exponential decay of correlations~\cite{Brandao12}, or for the study of black hole radiation~\cite{HayPre07,braunstein-pati,braunstein-zyczkowski}.

Information-theoretic applications other than state merging (cf.~Section~\ref{sec:applications}) have been investigated in the doctoral thesis of one of the authors~\cite{fred-these}. One of these applications is channel coding. Here, Alice wants to use a noisy quantum channel $\mathcal{N}^{A \rightarrow B}$ to send qubits to Bob with fidelity at least $1-\varepsilon$. The idea is that decoding is possible whenever a purification of the qubits Alice is sending is decoupled from the channel environment. One can therefore get a coding theorem directly from Theorem \ref{thm:achievability} by setting $\mathcal{T}$ to be the complementary channel of $\mathcal{N}$ (i.e., consider a Stinespring dilation~\cite{Stinespring55} $\cU^{\mathcal{N}}_{A \rightarrow BE}$ of $\mathcal{N}$, and set $\mathcal{T}_{A \rightarrow E}(\cdot)=\tr_B[U_{A}\cdot U_{A}\mdag]$). Unassisted channel coding~\cite{lsd1,lsd2,lsd3} can be obtained by choosing the input state $\rho_{AR} = \Phi_{AR}$ (where $\Phi_{A R}$ is a maximally entangled state between $A$ and $R$). Similarly, entanglement-assisted channel coding \cite{BSST02} corresponds to the input choice $\rho_{ABR} = \Phi_{A_R R} \otimes \Phi_{A_B B}$ (where $\cH_A=\cH_{A_R} \otimes \cH_{A_B}$, with $A_R$ containing the state to be transmitted and $A_B$ the initial entanglement that Alice shares with Bob). Other choices of $\rho_{ABR}$ correspond to different scenarios.

Another application where decoupling can be employed as a building block for constructing protocols is the simulation of noisy quantum channels using perfect classical channels together with pre-shared entanglement. The fully quantum reverse Shannon theorem asserts that this is possible using only a classical communication rate equal to the capacity of the channel to be simulated~\cite{BSST02,BDHSW09}. In~\cite{BeChRe09_2}, a proof of this theorem using one-shot decoupling has been proposed.

Our one-shot decoupling results contrast with (and are strictly more general than) the iid scenario\footnote{The abbreviation iid stands for independent and identically distributed.} usually considered in information theory, where statements are proved asymptotically under the assumption that the underlying processes (such as channel uses) are repeated many times independently. We note that asymptotic iid statements can be easily retrieved from the general one-shot results using the quantum asymptotic equipartition property (AEP) for smooth entropies~\cite{renner-phd,ToCoRe09} (see Lemma~\ref{thm:fully-quantum-aep}). Consider decoupling with a map of the form $\mathcal{\bar{T}} =\mathcal{T} \circ \mathcal{U}$ (with $\cU$ a random unitary channel). If the map $\cT$ as well as the initial state $\rho_{AE}$ consist of many identical copies, i.e., $\cT^{\otimes n}$ and $\rho^{\otimes n}_{AE}$, then the achievability bound of Theorem~\ref{thm:achievability}, i.e., the condition that is sufficient for decoupling, turns into the criterion
\begin{align}\label{eq:asymach}
H(A|E)_{\rho}+H(A|B)_{\tau} \geq 0 \ ,
\end{align}
where $H$ denotes the (conditional) von Neumann entropy. Analogously, the converse in Corollary~\ref{cor:conv} (i.e., the condition which is necessary for decoupling for maps of this form) turns into 
\begin{align}\label{eq:asymconv}
H(A|E)_{\rho}+H(A|B)_{\tau} \leq 0\ .
\end{align}
In other words, in the iid scenario, the achievability bound~\eqref{eq:asymach} and the converse bound \eqref{eq:asymconv}, taken together, imply an exact characterization of decoupling.

%%%%%%%%%%%%%%%%%%%%%%%%%%%%%%%%%%%%%%%%%%%%%%%%%%%%%%%%%%%%%%%%%%%%%

\section*{Acknowledgments}

We thank Andreas Winter for insightful discussions and for his valuable contributions to~\cite{diploma-berta}, which served as a starting point for this work. We also thank Patrick Hayden for enlightening discussions, as well as Oleg Szehr for fixing a bug regarding smoothing in the proof of Theorem~\ref{thm:achievability}, among other useful comments. We acknowledge support from the Swiss National Science Foundation (grants No.~200021-119868 and 200020-135048), the National Centre of Competence in Research 'Quantum Science and Technology (QSIT)', and the European Research Council (grant No.~258932). FD was supported by Canada’s NSERC Postdoctoral Fellowship Program. MB was supported by the German Science Foundation (grant CH 843/2-1), the Swiss National Science Foundation (grants PP00P2-128455, 20CH21-138799 (CHIST-ERA project CQC)), and the Swiss State Secretariat for Education and Research supporting COST action MP1006. JW was funded by the U.K. EPSRC grant EP/E04297X/1 and the Canada-France NSERC-ANR project FREQUENCY. Parts of this work were done while JW was at the University of Bristol.

\bibliographystyle{alpha}
\bibliography{big6}

%%%%%%%%%%%%%%%%%%%%%%%%%%%%%%%%%%%%%%%%%%%%%%%%%%%%%%%%%%%%%%%%%%%%%

\appendix

\section{Properties of Smooth Entropies}\label{app_smooth}

The conditional collision entropy is lower bounded by the conditional min-entropy.

\begin{lem}\label{lem:h2-hmin}
Let $\rho_{AB} \in \cS_{\leqslant}(\cH_{AB})$. Then, we have that $H_2(A|B)_{\rho} \geqslant H_{\min}(A|B)_{\rho}$.
\end{lem}

\begin{proof}
Let $\sigma_B\in\cS_{=}(\cH_{B})$ be such that $\rho_{AB} \leqslant2^{-H_{\min}(A|B)_{\rho}}\cdot\ident_A \otimes \sigma_B$. We then obtain
\begin{align}
2^{-H_2(A|B)_{\rho}} &= \min_{\omega_B} \tr\left[ (\ident_A \otimes \omega_B)^{-1/2} \rho_{AB} (\ident_A \otimes \omega_B)^{-1/2} \rho_{AB} \right]\nonumber\\
&\leqslant \tr\left[ (\ident_A \otimes \sigma_B)^{-1/2} \rho_{AB} (\ident_A \otimes \sigma_B)^{-1/2} \rho_{AB} \right]\nonumber\\
&\leqslant 2^{-H_{\min}(A|B)_{\rho}}\cdot\tr\left[ \ident_{AB} \rho_{AB} \right]\nonumber\\
&\leqslant 2^{-H_{\min}(A|B)_{\rho}} \ .
\end{align}
\end{proof}

The smooth conditional min-entropy is superadditive.

\begin{lem}\label{super}
Let $\eps,\eps'\geq0$, $\rho_{AB}\in\cS_{=}(\cH_{AB})$ and $\rho_{A'B'}'\in\cS_{=}(\cH_{A'B'})$. Then, we have that
\begin{align}
H_{\min}^{\eps+\eps'}(AA'|BB')_{\rho\otimes\rho'}\geq H_{\min}^{\eps}(A|B)_{\rho}+H_{\min}^{\eps'}(A'|B')_{\rho'}\ .
\end{align}
\end{lem}

\begin{proof}
Let $\bar{\rho}_{AB}\in\cB^{\eps}(\rho_{AB})$ and $\bar{\rho}'_{A'B'}\in\cB^{\eps}(\rho'_{A'B'})$ such that $H_{\min}^{\eps}(A|B)_{\rho}=H_{\min}(A|B)_{\bar{\rho}}$ and $H_{\min}^{\eps'}(A'|B')_{\rho'}=H_{\min}(A'|B')_{\bar{\rho}'}$. By the triangle inequality for the purified distance~\cite[Lemma 5]{duality-min-max-entropy} we have $\bar{\rho}_{AB}\otimes\bar{\rho}'_{A'B'}\in\cB^{\eps+\eps'}(\rho_{AB}\otimes\rho_{A'B'}')$. Using the additivity of the conditional min-entropy~\cite{min-max-entropy}, we conclude
\begin{align}
H_{\min}^{\eps+\eps'}(AA'|BB')_{\rho\otimes\rho'}&\geq H_{\min}(AA'|BB')_{\bar{\rho}\otimes\bar{\rho}'}\nonumber\\
&=H_{\min}(A|B)_{\bar{\rho}}+H_{\min}(A'|B')_{\bar{\rho}'}\nonumber\\
&=H_{\min}^{\eps}(A|B)_{\rho}+H_{\min}^{\eps'}(A'|B')_{\rho'}\ .
\end{align}
\end{proof}

We have the following dimension lower and upper bounds for the (smooth) conditional min-entropy.

\begin{lem}\cite[Lemma 20]{duality-min-max-entropy}\label{lem:dim}
Let $\rho_{AB}\in\cS_{=}(\cH_{AB})$. Then, we have that $H_{\min}(A|B)_{\rho}\geq-\log|B|$.
\end{lem}

\begin{lem}\label{old}
Let $\eps\geq0$ and $\rho_{ABC}\in\cS_{=}(\cH_{ABC})$. Then, we have that
\begin{align}
H_{\min}^{\eps}(AB|C)_{\rho}\leq H_{\min}^{\eps}(A|C)_{\rho}+\log|B|\ .
\end{align}
\end{lem}

\begin{proof}
Let $\bar{\rho}_{ABC}\in\cB^{\eps}(\rho_{ABC})$, $\sigma_{C}\in\cS_{=}(\cH_{C})$ and $\lambda\in\mathbb{R}$ such that
\begin{align}
H_{\min}^{\eps}(AB|C)_{\rho}=H_{\min}(AB|C)_{\bar{\rho}}=-\log\lambda\ ,
\end{align}
that is, $\lambda$ is minimal such that $\lambda\cdot\ident_{AB}\otimes\sigma_{C}-\bar{\rho}_{ABC}\geq0$. By taking the partial trace over $B$ we get $\lambda\cdot|B|\cdot\ident_{A}\otimes\sigma_{C}-\bar{\rho}_{AC}\geq0$. Furthermore we have by the monotonicity of the purified distance~\cite[Lemma 7]{duality-min-max-entropy} that $\bar{\rho}_{AC}\in\cB^{\eps}(\rho_{AC})$ and hence
\begin{align}
H_{\min}^{\eps}(A|C)_{\rho}\geq H_{\min}(A|C)_{\bar{\rho}}\geq-\log\mu\ ,
\end{align}
where $\mu\in\mathbb{R}$ is minimal such that $\mu\cdot\ident_{A}\otimes\sigma_{C}-\bar{\rho}_{AC}\geq0$. 
Thus $\lambda\cdot|B|\geq\mu$ and therefore
\begin{align}
H_{\min}^{\eps}(AB|C)_{\rho}\leq H_{\min}^{\eps}(A|C)_{\rho}+\log|B|\ .
\end{align}
\end{proof}

The following lemma is about the conditional min-entropy of quantum-classical states.

\begin{lem}\label{classical}
Let $\rho_{ABX}\in\cS_{=}(\cH_{ABX})$ with $\rho_{ABX}=\sum_{x}p_{x}\cdot\rho_{AB}^{x}\otimes\ket{x}\bra{x}_{X}$ and $\rho_{AB}^{x}\in\cS_{=}(\cH_{AB})$ for all $x$. Then, we have that
\begin{align}
H_{\min}(A|BX)_{\rho}=-\log(\sum_{x}p_{x}\cdot2^{-H_{\min}(A|B)_{\rho^{x}}})\ .
\label{oper}
\end{align}
\end{lem}

\begin{proof}
By the operational interpretation of the conditional min-entropy as the maximal achievable singlet fraction~\cite[Theorem 2]{min-max-entropy} we have
\begin{align}
H_{\min}(A|BX)_{\rho}=-\log(|A|\cdot\max_{\mathcal{F}_{BX\rightarrow A'}}F^{2}((\opid_{A}\otimes\mathcal{F}_{BX\rightarrow A'})(\rho_{ABX}),\ket{\Phi}\bra{\Phi}_{AA'}))\ ,
\end{align}
where the maximum is taken over all TPCPMs $\mathcal{F}_{BX\rightarrow A'}$, $\ket{\Phi}_{AA'}=|A|^{-1/2}\sum_{i}\ket{x}_{A}\otimes\ket{x}_{A'}$, and $\mathcal{H}_{A'}\cong\mathcal{H}_{A}$. Writing out the conditional min-entropy terms on the right hand side of~\eqref{oper} in the same manner we obtain
\begin{align}
H_{\min}(A|B)_{\rho^{x}}=-\log\left(|A|\cdot\max_{\mathcal{F}^{x}_{B\rightarrow A'}}F^{2}((\opid_{A}\otimes\mathcal{F}^{x}_{B\rightarrow A'})(\rho_{AB}^{x}),\ket{\Phi}\bra{\Phi}_{AA'})\right)\ .
\end{align}
The claim is therefore equivalent to
\begin{align}
&\max_{\mathcal{F}_{BX\rightarrow A'}}F^{2}((\opid_{A}\otimes\mathcal{F}_{BX\rightarrow A'})(\rho_{ABX}),\ket{\Phi}\bra{\Phi}_{AA'})\nonumber\\
&=\sum_{x}p_{x}\cdot\max_{\mathcal{F}^{x}_{B\rightarrow A'}}F^{2}((\opid_{A}\otimes\mathcal{F}^{x}_{B\rightarrow A'})(\rho_{AB}^{x}),\ket{\Phi}\bra{\Phi}_{AA'})\ .
\end{align}
Now, because the state $\rho_{ABX}$ is classical on $X$, the maximization on the left hand side can without loss of generality be restricted to TPCPMs that first measure on $X$ in the basis $\{\ket{x}\}$ and then do some TPCPM $\mathcal{F}^{x}_{B\rightarrow A'}$ conditioned on the measurement outcome $x$.  By the linearity of the square of the fidelity when one argument is pure, the claim then follows.
\end{proof}

The following lemma is about the conditional max-entropy of quantum-classical states.

\begin{lem}\label{classical2}
Let $\rho_{ABX}\in\cS_{=}(\cH_{ABX})$ with $\rho_{ABX}=\sum_{x}p_{x}\cdot\rho_{AB}^{x}\otimes\ket{x}\bra{x}_{X}$ and $\rho_{AB}^{x}\in\cS_{=}(\cH_{AB})$ for all $x$. Then, we have that
\begin{align}
H_{\max}(A|BX)_{\rho}=\log(\sum_{x}p_{x}\cdot2^{H_{\max}(A|B)_{\rho^{x}}})\ .
\end{align}
\end{lem}

\begin{proof}
Let $\rho_{ABCXX'}$ be a purification of $\rho_{ABX}$. Then, we have by the duality of conditional min- and max-entropy (Lemma~\ref{lem:dual}) and a lemma about the conditional min-entropy of quantum-classical states (Lemma~\ref{classical}) that
\begin{align}
H_{\max}(A|BX)_{\rho}=-H_{\min}(A|CX')_{\rho}=\log(\sum_{x}p_{x}\cdot2^{-H_{\min}(A|C)_{\rho^{x}}})=\log(\sum_{x}p_{x}\cdot2^{H_{\max}(A|B)_{\rho^{x}}})\ .
\end{align}
\end{proof}

The following lemma is property of the smooth conditional min-entropy of quantum-classical states.

\begin{lem}\label{lem:cc_equal}
Let $\varepsilon\geq0$ and $\rho_{ABXX'}\in\cS_{=}(\cH_{ABXX'})$ with $\rho_{ABXX'}=\sum_{x}p_{x}\cdot\rho_{AB}^{x}\otimes\ket{x}\bra{x}_{X}\otimes\ket{x}\bra{x}_{X'}$ and $\rho_{AB}^{x}\in\cS_{=}(\cH_{AB})$ for all $x$. Then, we have that
\begin{align}
H_{\min}^{\varepsilon}(A|BX)_{\rho}=H_{\min}^{\varepsilon}(AX'|BX)_{\rho}\ .
\end{align}
\end{lem}

\begin{proof}
We first show the case $\varepsilon=0$. By a property of the conditional min-entropy of quantum-classical states (Lemma~\ref{classical}), the claim becomes equivalent to
\begin{align}
H_{\min}(A|B)_{\rho^{x}}=H_{\min}(AX'|B)_{\rho^{x}\otimes\ket{x}\bra{x}}\ .
\end{align}
But by the additivity of the conditional min-entropy~\cite{min-max-entropy} this holds.

For $\varepsilon>0$, let $\bar{\rho}_{ABXX'}\in\cB^{\varepsilon}(\rho_{ABXX'})$ be classical on $XX'$ with respect to the basis $\{\ket{x}\otimes\ket{x}\}_{x}$ such that $H_{\min}^{\varepsilon}(AX'|BX)_{\rho}=H_{\min}(AX'|BX)_{\bar{\rho}}$ (which is possible by~\cite[Proposition 5.8]{Tomamichel12}). Since the purified distance is monotone under trace non-increasing CPMs~\cite[Lemma 7]{duality-min-max-entropy}, we have $\bar{\rho}_{ABX}\in\cB^{\varepsilon}(\rho_{ABX})$ and hence
\begin{align}
H_{\min}^{\varepsilon}(AX'|BX)_{\rho}\leq H_{\min}^{\varepsilon}(A|BX)_{\rho}\ .
\end{align}
For the inequality in the other direction, let $\hat{\rho}_{ABX}\in\cB^{\varepsilon}(\rho_{ABX})$ be classical on $X$ with respect to the basis $\{\ket{x}\}_{x}$ such that $H_{\min}^{\varepsilon}(A|BX)_{\rho}=H_{\min}(A|BX)_{\hat{\rho}}$ (which is possible by~\cite[Proposition 5.8]{Tomamichel12}). By~\cite[Corollary 9]{duality-min-max-entropy} and the monotonicity of the purified distance under trace non-increasing CPMs~\cite[Lemma 7]{duality-min-max-entropy} there exists an extension $\hat{\rho}_{ABXX'}\in\cB^{\varepsilon}(\rho_{ABXX'})$ of $\hat{\rho}_{AXB}$ that is classical on $XX'$ with respect to the basis $\{\ket{x}\otimes\ket{x}\}_{x}$. Thus, we conclude
\begin{align}
H_{\min}^{\varepsilon}(A|BX)_{\rho}\leq H_{\min}^{\varepsilon}(AX'|BX)_{\rho}\ .
\end{align}
\end{proof}

We have the following chain rule for the smooth conditional min-entropy.

\begin{lem}\label{lem:chainrule}
Let $\eps>0$, $\eps',\eps''\geq0$ and $\rho_{ABC}\in\cS_{=}(\cH_{ABC})$. Then, we have that
\begin{align}
H_{\min}^{\eps + 2\eps' + \eps''}(AB|C)_{\rho} \geqslant H_{\min}^{\eps'}(A|BC)_{\rho} + H_{\min}^{\eps''}(B|C)_{\rho} - \log \frac{2}{\eps^2}\ .
\end{align}
\end{lem}

\begin{proof}
Let $\rho_{ABC}'\in\cB^{\eps'}(\rho_{ABC})$ such that $H_{\min}^{\eps'}(A|BC)_{\rho}=H_{\min}(A|BC)_{\rho'}$ and let $\rho_{ABCE}'$ be a purification of $\rho_{ABC}'$. Furthermore let $\rho_{BC}''\in\cB^{\eps''}(\rho_{BC})$, $\sigma_{C}\in\cS_{=}(\cH_{BC})$ and $\lambda\in\mathbb{R}$ such that $H_{\min}^{\eps''}(B|C)_{\rho}=H_{\min}(B|C)_{\rho''}=-\log\lambda$, that is, $\lambda$ is minimal such that
\begin{align}\label{AB}
\lambda\cdot\ident_{B}\otimes\sigma_{C}-\rho_{BC}''\geq0\ .
\end{align}
By~\cite[Lemma 21]{leftover} there exists a projector $P_{AE}$ such that
\begin{align}
\bar{\rho}'_{ABCE}=(P_{AE}\otimes\ident_{BC})\rho'_{ABCE}(P_{AE}\otimes\ident_{BC})\in\cB^{\eps}(\rho_{ABCE}')\ ,
\end{align}
and
\begin{align}\label{ABC}
2^{-H_{\min}^{\eps'}(A|BC)_{\rho}+\log\frac{2}{\eps^{2}}}\cdot\ident_{A}\otimes\rho_{BC}'-\bar{\rho}_{ABC}'\geq0\ .
\end{align}
Now let $T_{BC}$ be defined as in Lemma~\ref{lem:extuhlm} with
$\rho_{BC}''=T_{BC}\rho_{BC}'T_{BC}^{\dagger}$ and consider the state
\begin{align}
\bar{\rho}_{ABCE}''=(\ident_{AE}\otimes T_{BC})\bar{\rho}_{ABCE}'(\ident_{AE}\otimes T_{BC}^{\dagger})=(P_{AE}\otimes T_{BC})\rho_{ABCE}'(P_{AE}\otimes T_{BC}^{\dagger})\ .
\end{align}
Applying $T_{BC}$ to~\eqref{ABC} we obtain
\begin{align}
2^{-H_{\min}^{\eps'}(A|BC)_{\rho}+\log\frac{2}{\eps^{2}}}\cdot\ident_{A}\otimes\rho_{BC}''-\bar{\rho}_{ABC}''\geq0\ .
\end{align}
Together with~\eqref{AB} this yields
\begin{align}
2^{-H_{\min}^{\eps'}(A|BC)_{\rho}+\log\frac{2}{\eps^{2}}-H_{\min}^{\eps''}(B|C)_{\rho}}\cdot\ident_{AB}\otimes\sigma_{C}-\bar{\rho}_{ABC}''\geq0\ .
\end{align}
This implies
\begin{align}\label{ABCD}
H_{\min}(AB|C)_{\bar{\rho}''}\geq H_{\min}^{\eps'}(A|BC)_{\rho}+H_{\min}^{\eps''}(B|C)_{\rho}-\log\frac{2}{\eps^{2}}\ .
\end{align}
But by the monotonicity of the purified distance~\cite[Lemma 7]{duality-min-max-entropy} and the definition of $T_{BC}$ we have
\begin{align}
P(\bar{\rho}_{ABC}'',\bar{\rho}_{ABC}')&\leq P((P_{AE}\otimes T_{BC})\rho_{ABCE}'(P_{AE}\otimes T_{BC}^{\dagger}),(P_{AE}\otimes\ident_{BC})\rho_{ABCE}'(P_{AE}\otimes\ident_{BC}))\nonumber\\
&\leq P((\ident_{AE}\otimes T_{BC})\rho_{ABCE}'(\ident_{AE}\otimes T_{BC}^{\dagger}),\rho_{ABCE}')\nonumber\\
&=P(\rho_{BC}'',\rho_{BC}')\ ,
\end{align}
and hence
\begin{align}
P(\bar{\rho}_{ABC}'',\bar{\rho}_{ABC}')\leq P(\rho_{BC}'',\rho_{BC})+P(\rho_{BC},\rho_{BC}')\leq\eps''+\eps'\ .
\end{align}
Finally we obtain
\begin{align}
P(\bar{\rho}''_{ABC},\rho_{ABC})&\leq P(\bar{\rho}''_{ABC},\bar{\rho}_{ABC}')+P(\bar{\rho}'_{ABC},\rho_{ABC}')+P(\rho_{ABC}',\rho_{ABC})\nonumber\\
&\leq\eps''+\eps'+\eps+\eps'=\eps + 2\eps' + \eps'\ ,
\end{align}
and thus together with~\eqref{ABCD} that
\begin{align}
H_{\min}^{\eps + 2\eps' + \eps'}(AB|C)_{\rho}\geq H_{\min}^{\eps'}(A|BC)_{\rho}+H_{\min}^{\eps''}(B|C)_{\rho}-\log\frac{2}{\eps^{2}}\ .
\end{align}
\end{proof}

%%%%%%%%%%%%%%%%%%%%%%%%%%%%%%%%%%%%%%%%%%%%%%%%%%%%%%%%%%%%%%%%%%%%%

\section{Technical Lemmas}

\begin{lem}\cite[Lemma 6]{duality-min-max-entropy}\label{lem:purified-dist-vs-trace-dist}
Let $\rho, \sigma \in \cS_{\leqslant}(\cH)$. Then, we have that
\begin{align}
\bar{D}(\rho, \sigma) \leqslant P(\rho, \sigma) &\leqslant \sqrt{2\bar{D}(\rho,\sigma)} \leqslant \sqrt{2 \| \rho - \sigma \|_1}\\
\frac{1}{2}P(\rho,\sigma)^2 \leqslant \bar{D}(\rho,\sigma) &\leqslant P(\rho,\sigma)\ ,
\end{align}
where $\bar{D}(\rho,\sigma) = \frac{1}{2}\| \rho - \sigma \|_1 + \frac{1}{2}| \tr[\rho] - \tr[\sigma]|$.
\end{lem}

\begin{lem}\label{lem:hmax-optimal-z}
Let $\rho_{ABC} \in \cS_{\leqslant}(\cH_{ABC})$ be pure. Then, we have that for any $\sigma_{B}\in\cS_{=}(\cH_{B})$ with full rank,
\begin{align}
\rho_{ABC} \leqslant Z_{AB} \otimes \ident_C\ ,
\end{align}
where $Z_{AB} = 2^{\demi \hmax(A|B)_{\rho | \sigma}}\cdot\sigma_B^{-1/2} \sqrt{\sigma_B^{1/2} \rho_{AB} \sigma_B^{1/2}} \sigma_B^{-1/2}$. Furthermore, $Z_{AB}$ has the property that $\tr[Z_{AB} \sigma_B] = 2^{\hmax(A|B)_{\rho|\sigma}}$.
\end{lem}

\begin{proof}
Consider the following semidefinite program (for a introduction to semidefinite programs presented in this manner, see for instance \cite{watrousnotes}):
\begin{center}
  \begin{minipage}{6cm}
    \centerline{\underline{Primal}}\vspace{-7mm}
    \begin{align}
      \text{maximize:}\quad & \tr[\rho_{ABC} X_{ABC}]\nonumber\\
      \text{subject to:}\quad & \tr_C[X_{ABC}] = \ident_A \otimes \sigma_B\nonumber\\
      & X_{ABC} \geqslant 0\nonumber
    \end{align}
  \end{minipage}
  \hspace*{10mm}
  \begin{minipage}{6cm}
    \centerline{\underline{Dual}}\vspace{-7mm}
    \begin{align}
      \text{minimize:}\quad & \tr[(\ident_A \otimes \sigma_B) Z_{AB}]\nonumber\\
    \text{subject to:}\quad & \rho_{ABC} \leqslant Z_{AB} \otimes \ident_C\nonumber\ .
    \end{align}
  \end{minipage}
\end{center}
From the definition of the conditional max-entropy (Definition~\ref{hmax-rho-given-sigma}) and Uhlmann's theorem~\cite{uhlmann} it is clear that the optimal value of the primal problem is $2^{\hmax(A|B)_{\rho|\sigma}}$. One can also easily show that strong duality holds (i.e., that the optimal value of the dual problem is equal to that of the primal problem). One simply needs to show that there exists a $Z_{AB}$ such that $Z_{AB} \otimes \ident_C > \rho_{ABC}$, which holds for $Z_{AB}=2\cdot\ident_{AB}$.

Now, we need to show that the optimal $Z_{AB}$ for this problem has the form given in the lemma statement. First, note that by Uhlmann's theorem~\cite{uhlmann}, there must exist an optimal $X_{ABC}$ which has rank 1, assuming we consider the system $C$ to be large enough. Let $X_{ABC} = \proj{\varphi}_{ABC}$ and let $\rho_{ABC} = \proj{\rho}_{ABC}$, and consider the complementary slackness condition for $X$ and $Z$ to be optimal: $\rho_{ABC} X_{ABC} = (Z_{AB} \otimes \ident_C) X_{ABC}$. We can rewrite this as
\begin{align}
\braket{\rho}{\varphi} \ket{\rho}\bra{\varphi} = (Z_{AB} \otimes \ident_C) \proj{\varphi}\ ,
\end{align} 
and therefore
\begin{align}
\braket{\rho}{\varphi} \ket{\rho} = (Z_{AB} \otimes \ident_C) \ket{\varphi}\ ,
\end{align} 
as well as
\begin{align}
F(\rho, \varphi)^2 \proj{\rho} = (Z_{AB} \otimes \ident_C) \proj{\varphi} (Z_{AB} \otimes \ident_C)\ .
\end{align} 
Tracing out $C$ and using the fact that $F(\rho,\varphi)^2 = 2^{\hmax(A|B)_{\rho|\sigma}}$, we get
\begin{align}
2^{\hmax(A|B)_{\rho|\sigma}}\cdot\rho_{AB} = Z_{AB} (\ident_A \otimes \sigma_B) Z_{AB}\ .
\end{align} 
Now, conjugating both sides by $\sigma_B^{1/2}$ and taking square roots on both sides, we get that
\begin{align}
2^{\demi \hmax(A|B)_{\rho|\sigma}}\cdot\sqrt{\sigma_B^{1/2} \rho_{AB} \sigma_B^{1/2}} = \sigma_B^{1/2} Z_{AB} \sigma_B^{1/2}\ .
\end{align} 
If $\sigma_B$ has full rank, we get the expression for $Z_{AB}$ by conjugating both sides by $\sigma_B^{-1/2}$. Finally, the fact that $\tr[Z_{AB} \sigma_B] = 2^{\hmax(A|B)_{\rho|\sigma}}$ can simply be computed from the expression for $Z$.
\end{proof}

\begin{lem}\label{lem:extuhlm}
Let $\rho_{AB}\in\cS_{\leq}(\cH_{AB})$ and $\sigma_{A}\in\cS_{\leq}(\cH_{A})$. Then, there exists $T_{A}\in\cL(\cH_{A})$ with
\begin{align}
\sigma_{AB}=(T_{A}\otimes\ident_{B})\rho_{AB}(T_{A}^{\dagger}\otimes\ident_{B})\in\cS_{\leq}(\cH_{AB})
\end{align}
an extension of $\sigma_{A}$ such that
$P(\rho_{AB},\sigma_{AB})=P(\rho_{A},\sigma_{A})$.
\end{lem}

\begin{proof}
Define $X_{A}=\sigma_{A}^{\frac{1}{2}}\rho_{A}^{\frac{1}{2}}$ and polar decompose $X_{A}=V_{A}(X^{\dagger}_{A}X_{A})^{1/2}$. Furthermore define $T_{A}=\sigma_{A}^{\frac{1}{2}}V_{A}\rho_{A}^{-\frac{1}{2}}$, where the inverse is a generalized inverse.\footnote{For $M\in\cP$, $M^{-1}$ is a generalized inverse of $M$ if $MM^{-1}=M^{-1}M=\mathrm{supp}(M)=\mathrm{supp}(M^{-1})$, where $\mathrm{supp}(\cdot)$ denotes the support.} We have
\begin{align}
\tr_{B}((T_{A}\otimes\ident_{B})\rho_{AB}(T_{A}^{\dagger}\otimes\ident_{B}))=T_{A}\rho_{A}T_{A}^{\dagger}=\sigma_{A}^{\frac{1}{2}}V_{A}V_{A}^{\dagger}\sigma_{A}^{\frac{1}{2}}=\sigma_{A}\ ,
\end{align}
which shows that $\sigma_{AB}=(T_{A}\otimes\ident_{B})\rho_{AB}(T_{A}^{\dagger}\otimes\ident_{B})$ is an extension of $\sigma_{A}$. Thus it remains to prove that $P(\rho_{AB},\sigma_{AB})=P(\rho_{A},\sigma_{A})$.

For this we first assume that $\rho_{AB}$ is pure and normalized, i.e., $\rho_{AB}=\ket{\rho}\bra{\rho}_{AB}\in\cS_{=}(\cH_{AB})$. Then, we have that 
\begin{align}
P(\rho_{AB},\sigma_{AB})&=\sqrt{1-|\braket{\rho}{\sigma}|^2}\nonumber\\
&=\sqrt{1-\left|\tr\left[(T_A \otimes \ident_B) \rho_{AB}\right]\right|^2}\nonumber\\
&=\sqrt{1-\left|\tr\left[(\sigma_A^{1/2} V_A \rho_A^{-1/2} \otimes \ident_B) \rho_{AB}\right]\right|^2}\nonumber\\
&=\sqrt{1-\left|\tr\left[\sigma_A^{1/2} V_A \rho_A^{1/2}\right]\right|^2}\nonumber\\
&=\sqrt{1-\left|\tr\left[\rho_A^{1/2}\sigma_A^{1/2} V_A\right]\right|^2}\nonumber\\
&=\sqrt{1-\left|\tr\left[\sqrt{\rho_A^{1/2} \sigma_A \rho_A^{1/2}}\right]\right|^2}\nonumber\\
&=\sqrt{1-F^{2}(\rho_{A},\sigma_{A})}\nonumber\\
&=P(\rho_{A},\sigma_{A})\ .
\end{align}
If $\rho_{AB}=\ket{\rho}\bra{\rho}_{AB}$ is not normalized we obtain analogously
\begin{align}
P(\rho_{AB},\sigma_{AB})&=\sqrt{1-[F(\rho_{AB},\sigma_{AB})+\sqrt{(1-\tr[\rho_{AB}])(1-\tr[\sigma_{AB}])}]^{2}}\nonumber\\
&=\sqrt{1-\left(F(\rho_{A},\sigma_{A})+\sqrt{(1-\tr[\rho_{A}])(1-\tr[\sigma_{A}])}\right)^{2}}\nonumber\\
&=P(\rho_{A},\sigma_{A})\ .
\end{align}
The statement for a general $\rho_{AB}$ (not necessarily pure) follows by the monotonicity of the purified distance~\cite[Lemma 7]{duality-min-max-entropy} under partial trace.
\end{proof}

\end{document}